\documentclass[orivec]{llncs}

\usepackage{
      afterpage,   
      amsmath,     
      mathtools,   
      amssymb,     
      amsfonts,    
      pdfsync,      
      url,          
      xspace,       
      }   

\usepackage{dsfont} 
\usepackage{wrapfig} 
\usepackage{stmaryrd} 
 \usepackage[inline]{enumitem} 
\usepackage{centernot}
\usepackage{hyperref}
\usepackage{cleveref}
\crefname{figure}{Fig.}{Figs.}
\crefname{definition}{Def.}{Defs.}
\crefname{defn}{Def.}{Defs.}
\crefname{equation}{Eq.}{Eqs.}
\crefname{theorem}{Thm.}{Thms.}
\crefname{thm}{Thm.}{Thms.}
\crefname{lemma}{Lem.}{Lems.}
\crefname{lem}{Lem.}{Lems.}
\crefname{corollary}{Cor.}{Cors.}
\crefname{cor}{Cor.}{Cors.}
\crefname{proposition}{Prop.}{Props.}
\crefname{prop}{Prop.}{Props.}
\crefname{enumi}{Item}{Items}
\crefname{example}{Example}{E.g.}
\crefname{section}{Sec.}{Secs.}
\crefname{table}{Table}{Tables}
\crefname{appendix}{Appendix}{Apps.}
\crefname{lstlisting}{List.}{List.}
\Crefname{lstlisting}{List.}{List.}

\usepackage{amsthm}  
\theoremstyle{plain}

\theoremstyle{definition}

\theoremstyle{remark}

\theoremstyle{plain}
\newtheorem*{thm*}{Theorem}
\newtheorem*{lem*}{Lemma}
\newtheorem*{prop*}{Proposition}
\newtheorem*{cor*}{Corollary}

\theoremstyle{definition}
\newtheorem*{defn*}{Definition}
\newtheorem*{conj*}{Conjecture}
\newtheorem*{exmp*}{Example}

\newenvironment{thmbiss}[1]{\begin{trivlist}
\item[\hskip \labelsep {\bfseries Theorem #1}]
\it}{\end{trivlist}}

\usepackage{tikz}  
\usetikzlibrary{matrix,shapes,arrows}
\tikzstyle{block} = [rectangle, draw=black,
fill=blue!20,
text centered, minimum height=1.5em]
\tikzstyle{line} = [draw, thick, color=black!90, -latex']

\usepackage[inference]{semantic} 

\usepackage{pic}  
\usepackage{monitor} 
\usepackage{logic}  
\usepackage{shorthand}  
\usepackage{mathpartir}  

\newcommand{\qedmaybe}{}
\newcommand{\subsectionmaybe}[1]{\paragraph{#1}}

\usepackage{comment}
\includecomment{comment}
\specialcomment{karoliina}{\begingroup\sffamily\color{blue}}{\endgroup}
\specialcomment{antonis}{\begingroup\sffamily\color{magenta}}{\endgroup}	
\specialcomment{adrian}{\begingroup\sffamily\color{green}}{\endgroup}

\newcommand{\diam}[1]{\langle #1 \rangle}

\newcommand{\mx}{\textsf{max}~} 

\newcommand{\mend}{\textsf{end}}

\newcommand{\succeed}{\texttt{success}}

\newcommand{\fa}{\textsf{f}\xspace}

\newcommand{\su}{\textsf{s}\xspace}
\newcommand{\re}{\textsf{r}\xspace}

\usepackage[noadjust]{cite}

\begin{document}

\title{An Operational Guide to Monitorability}
\author{L. Aceto\inst{1}\inst{2}\and  A. Achilleos\inst{2} \and A. Francalanza\inst{3}\and A. Ing\'{o}lfsd\'{o}ttir\inst{2}\and K. Lehtinen\inst{4}}

\institute{Gran Sasso Science Institute, L'Aquila, Italy
\and
Reykjavik University, Reykjavik, Iceland
\and
University of Malta, Msida,  Malta
\and
University of Liverpool, Liverpool, UK}

\maketitle

\begin{abstract}
Monitorability
delineates what properties can be verified at runtime.
Although many monitorability definitions exist, few are defined explicitly in terms of the guarantees provided by monitors, \ie the computational entities carrying out the verification.
%
%
%
We 
view monitorability as a spectrum:
the fewer monitor guarantees that are required, the more properties become monitorable.
%
 We present a monitorability hierarchy and provide operational and syntactic characterisations for its levels.
%
Existing monitorability definitions are mapped into our hierarchy, providing a unified framework that makes the operational assumptions and guarantees of each definition explicit.
This  provides a rigorous foundation that can inform design choices and correctness claims for runtime verification tools.
%
%
\end{abstract}

\section{Introduction}
\label{sec:introduction}

 Any sufficiently expressive specification language contains properties that cannot be monitored at runtime~\cite{MannaPnueli:91:TCS,PnueliZaks:06:FM,Diekert2014,CiniFrancalanza:15:TACAS,FraAI:17:FMSD,AceAFI:18:FOSSACS,AcetoAFIL19}.
For instance, the \emph{satisfaction} of a safety property (``bad things never happen") cannot, in general, be determined by observing the (finite) behaviour of a program up to the current execution point; its \emph{violation}, however, can.
\emph{Monitorability}~\cite{BartocciFFR:18:RVIntro} concerns itself with the delineation between properties that are monitorable and those that are not.
%
Monitorability is paramount for a slew of  Runtime Verification (RV) tools, such as those described in
\cite{chen-rosu-2007-oopsla,jUnitRv:13,RegerCR15,Attard:16:RV,NeykovaBY17} (to name but a few), that synthesise monitors from specifications expressed in a variety of logics.
These  monitors are then executed with the system under scrutiny to produce verdicts concerning the satisfaction or violation of the specifications from which they were synthesised.

Monitorability is crucial for a principled approach because it disciplines the construction of RV tools.
%
It defines, either explicitly or implicitly, a notion of \emph{monitor correctness}, which guides the automated synthesis.
It also delimits the
monitorable fragment of the specification logic on which the synthesis is defined: monitors need not be synthesised for non-monitorable specifications.
%
%
In some settings, a
syntactic characterisation of monitorable properties can be identified~\cite{FraAI:17:FMSD,AceAFI:17:FSTTCS,AcetoAFIL19}, and used as a \emph{core calculus} for studying 
optimisations of the synthesis algorithm.
%
%
%
%
%
%
More broadly, monitorability boundaries may
guide the design of \textit{hybrid} verification strategies, which combine RV with other verification techniques (see the work in \cite{AceAFI:18:FOSSACS} for an example of this approach).
%
%

In spite of its importance, there is \emph{no} generally accepted notion of monitorability to date.
The literature contains a number of definitions, such as the ones proposed in \cite{FalconeFernandezMounier:STTT:12,PnueliZaks:06:FM,bauer2011runtime,FraAI:17:FMSD,AcetoAFIL19}.
These differ in aspects such as
the adopted specification formalism,
\eg LTL, Street automata, \UHML \etc, the operational model, \eg testers, automata, process calculi \etc,  and the semantic domain, \eg infinite traces, finite and infinite (finfinite) traces or labelled transition systems.
Even after these differences are normalised, many of these definitions are \emph{not} in agreement: there are properties that are monitorable according to some definitions but \emph{not} monitorable according to others.
%
More alarmingly, as we will show, frequently cited definitions of monitorability
contain serious errors.

This
discrepancy between definitions raises the question of \emph{which one} to adopt to inform one's implementation of an RV tool, and \emph{what effect} this choice has on the behaviour of the resulting tool.
A difficulty in informing this choice is that few of those definitions make explicit the relationship between the operational model, \ie the behaviour of a monitor, and the monitored properties.
In other words, it is not clear what the guarantees provided by the various monitors mentioned in the literature are, and how they differ from each other.


\begin{example}\label{ex:intro}
Consider the runtime verification of a system exhibiting (only) three events over finfinite traces: failure (\fa), success (\su) and recovery (\re).
One property we may require is that \emph{``failure never occurs and eventually success is reached''}, otherwise expressed in LTL fashion as
\begin{math}
  (\ltlG\,\neg\fa)\wedge(\ltlF\,\su)
\end{math}.
According to the definition of monitorability attributed to Pnueli and Zaks~\cite{PnueliZaks:06:FM}
(discussed in
\Cref{sec:pz}), this property is monitorable. However, it is not monitorable according to others, including Schneider~\cite{schneider2000enforceable},  Viswanathan and Kim~\cite{viswanathan2004foundations}, and Aceto \etal~\cite{AcetoAFIL19},  whose definition of monitorability coincides with some subset of \textit{safety properties}.  \exqed
\end{example}

\paragraph{Contributions.}
To our mind, this state of the art is unsatisfactory for tool construction.
An RV tool broadly
relies on the following ingredients:
$(i)$ the \emph{input} of the tool in terms of
the formalism used to describe the specification properties;
$(ii)$ the executable description of monitors that are the tool's \emph{output} and $(iii)$ the \emph{mapping} between the inputs and outputs, \ie the synthesis function of monitors from specifications.
Any account on monitorability should, in our view, shed light on those three aspects,  particularly on what it means for the synthesis function and the monitors it produces to be \textit{correct}.
In addition, if this account is flexible enough to incorporate a variety of relationships between specification properties and the expected behaviour of monitors, it can then be used by the tool implementors as a principled foundation to guide their design decisions.

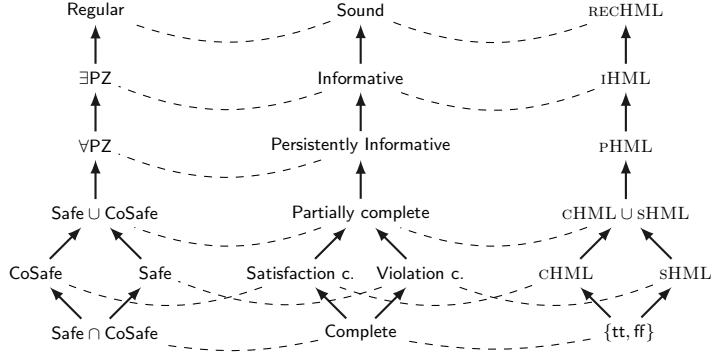
\begin{figure}[t]
      \vspace{-0.25cm}
    \centering
    \begin{tikzpicture}[>=latex,auto,scale=0.75, every node/.style={scale=0.75}]
     \begin{scope}[block, node distance=1.2cm]
       \node (fourFFM) {\clsReg}; 
       \node[below of = fourFFM] (epz)  {\clsEPZ};
       \node[below of = epz] (upz)  {\clsUPZ};
       \node[below of = upz] (threeFFM)  {$\quad\clsSafe\cup\clsCSafe$};
       \node[below left of = threeFFM, node distance=1.5cm] (twoFFMY)  {\clsCSafe};
       \node[below right of = threeFFM, node distance=1.5cm] (twoFFMN)  {\clsSafe};
       \node[below right of = twoFFMY, node distance=1.5cm] (Complete)  {$\quad\clsSafe\cap\clsCSafe$};
       \node[right of = fourFFM, xshift=3.5cm] (sound) {\lclsS};
       \node[below of = sound] (icmp)  {\lclsUC};
       \node[below of = icmp] (picmp)  {\lclsSUC};
       \node[below of = picmp] (pcmp)  {\lclsPC};
       \node[below left of = pcmp, node distance=1.5cm] (scmp)  {\lclsSC};
       \node[below right of = pcmp, node distance=1.5cm] (vcmp)  {\lclsVC};
       \node[below right of = scmp, node distance=1.5cm] (cmp)  {\lclsCC};
		\node[right of = sound, xshift=3.5cm] (rhml) {\UHML};
		\node[below of = rhml] (iml)  {\IHML};
		\node[below of = iml] (piml)  {\PIHML};
		\node[below of = piml] (mhml)  {$\CHML\cup\SHML$};
		\node[below left of = mhml, node distance=1.5cm] (chml)  {\CHML};
		\node[below right of = mhml, node distance=1.5cm] (shml)  {\SHML};
		\node[below right of = chml, node distance=1.5cm] (ttff)  {$~\{\hTru,\hFls \}$};
     \end{scope}

    \begin{scope}
     \path[line,<-] (sound) -- (icmp);
     \path[line,<-] (icmp) -- (picmp);
     \path[line,<-] (picmp) -- (pcmp);
     \path[line,<-] (pcmp) -- (scmp);
     \path[line,<-] (pcmp) -- (vcmp);
     \path[line,<-] (scmp) -- (cmp);
     \path[line,<-] (vcmp) -- (cmp);
      \path[line,<-] (fourFFM) -- (epz);
      \path[line,<-] (epz) -- (upz);
      \path[line,<-] (upz) -- (threeFFM);
      \path[line,<-] (threeFFM) -- (twoFFMY);
      \path[line,<-] (threeFFM) -- (twoFFMN);
      \path[line,<-] (twoFFMY) -- (Complete);
      \path[line,<-] (twoFFMN) -- (Complete);
	\path[line,<-] (rhml) -- (iml);
	\path[line,<-] (iml) -- (piml);
	\path[line,<-] (piml) -- (mhml);
	\path[line,<-] (mhml) -- (chml);
	\path[line,<-] (mhml) -- (shml);
	\path[line,<-] (chml) -- (ttff);
	\path[line,<-] (shml) -- (ttff);
    \end{scope}

    \path[dashed] (sound) edge [bend left=20]  (fourFFM)
    (icmp) edge [bend left=20]  (epz)
    (picmp) edge [bend left=20]  (upz)
    (pcmp) edge [bend left=20]  (threeFFM)
    (scmp) edge [bend left=20]  (twoFFMY)
    (vcmp) edge [bend left=20]  (twoFFMN)
    (cmp) edge [bend left=10]  (Complete);

    \path[dashed] (rhml) edge [bend left=20]  (sound)
    (iml) edge [bend left=20]  (icmp)
    (mhml) edge [bend left=20]  (pcmp)
    (chml) edge [bend left=20]  (scmp)
    (shml) edge [bend left=20]  (vcmp)
    (ttff) edge [bend left=10]  (cmp);
    \end{tikzpicture}
  \caption{The Monitorability Hierarchy of Regular Properties\label{fig:hierarchy}}
  \vspace{-0.25cm}
\end{figure}

For these reasons, we take the view that monitorability comes on a spectrum.
There is a trade-off between the guarantees provided by monitors and the properties that can be monitored with those guarantees.
We argue that considering different requirements gives rise to a \emph{hierarchy of monitorability}---depicted in \Cref{fig:hierarchy} (middle)---which classifies properties according to what types of guarantees RV can give for them.
At one extreme,
anything can be monitored if the only requirement is for monitors to be \emph{sound} \ie they should not contradict the monitored specification.
However, monitors that are \emph{just} sound give no guarantees of ever giving a verdict.
More usefully, \emph{\usefully} monitorable properties 
enjoy monitors that reach a verdict for \textit{some} finite execution; arguably, this is the minimum requirement for making monitoring potentially worthwhile.
More stringent requirements can demand this capability to be \emph{invariant} over monitor executions, \ie a monitor never reaches a state where it cannot provide a verdict; then we speak of persistently informative monitors.
%
%
Adding completeness requirements of different strengths, such as the requirement that a monitor should be able to identify all failures and/or satisfactions, yields stronger definitions of monitorability: partial, 
satisfaction or violation complete,
and complete.

In order not to favour a specific operational model, the hierarchy in \Cref{fig:hierarchy} (middle) is cast in terms of abstract behavioural requirements for monitors.
We then provide an instantiation that concretise those requirements into an \emph{operational} hierarchy, establishing
operational counterparts for each type of monitorability.
To this end, we use the operational framework developed in \cite{AcetoAFIL19}.
%
We show this framework to be, in a suitable technical sense, maximally general (\Cref{thm:determines-gives-verdict}). This shows that our work is equally applicable to other operational models.


In order for a tool to synthesise monitors from specifications, it is useful to have \emph{syntactic characterisations} of the properties that are monitorable with the required guarantees: synthesis can then directly operate on the syntactic fragment.
We
offer
monitorability characterisations as fragments of \UHML~\cite{Larsen:90:HMLRec,AceILS:2007} (a variant of the \UCalc~\cite{Koz:83:TCS}) interpreted over finfinite traces---see \Cref{fig:hierarchy} (right).
The logic is expressive enough to capture all regular properties---the focus of nearly all existing definitions of monitorability---and subsumes more user-friendly specification logics such as LTL.
%
Partial and complete monitorability already enjoy monitor synthesis functions and neat syntactic characterisations in \UHML~\cite{AcetoAFIL19}; related synthesis functions based on syntactic characterisations for a branching-time setting~\cite{FrancalanzaAAACDMI:17:RV,FraAI:17:FMSD} have already been implemented in a tool~\cite{Attard:16:RV,Attard:17:Book}.
Here, we provide the missing syntactic characterisation for \useful monitorability, and for a fragment of \superuseful monitorability.
%

Finally, we show that the proposed hierarchy accounts for existing notions of monitorability. See \Cref{fig:hierarchy} (left).
Safety, co-safety and their union correspond to partial monitorability and its two components, satisfaction- and violation-monitorability; 
Pnueli and Zaks's definition of monitorability can be interpretated in two ways, of which one (\epz) maps to \useful monitorability, and the other (\upz) to \superuseful monitorability.
%
We also show that the definitions of monitorability proposed by Falcone \etal~\cite{FalconeFernandezMounier:STTT:12}, contrary to their claim, do \emph{not} coincide with safety and co-safety properties.
%
To summarise, our principal contributions are:
\begin{enumerate}
\item
A unified operational perspective on existing notions of monitorability, clarifying what operational guarantees each provides, see \Cref{thm:mon-hierarchy,thm:safety-oper,thm:epz};
\item
An extension to the
syntactic characterisations of monitorable classes from  \cite{AcetoAFIL19}, mapping all but one of these classes to
fragments in \UHML, which can be viewed as a target byte-code for higher-level logics, see \Cref{thm:useful-is-ihml,thm:supuseful-regular-is-pihml}.
\end{enumerate}
Proofs omitted from the body of this paper can be found in \Cref{sec:proofs-appendix}.

\section{Preliminaries}
\label{sec:preliminaries}

\noindent\emph{Traces.}
We assume a finite set of actions, $\acta,\actb,\ldots {\in} \Act$.
%
%
%
The metavariables $\tV,\tVV{\in}\Act^\omega$ range over \emph{infinite} sequences of
actions.
%
\emph{Finite traces}, denoted as $\ftV,\ftVV \in \Act^\ast$, represent \emph{finite} prefixes of system runs.
Collectively, finite and infinite traces $\fTrc = \Act^\omega {\cup} \Act^\ast$ are called \emph{finfinite} traces. We use $\fftV,\fftVV \in \fTrc$ (\resp $\FSet \subseteq \fTrc$)
to range over  finfinite traces (\resp sets of finfinite traces).
%
A (finfinite) trace  with action \acta at its head is denoted as $\acta\fftV$.
Similarly, a (finfinite) trace  with a prefix \ftV\ and continuation \fftV\  is denoted as $\ftV\fftV$.
We write $\ftV \preceq \fftV$ to denote that the finite trace \ftV\ is a prefix of \fftV,
\ie
$\exists \fftVV \cdot \fftV = \ftV\fftVV$.

\medskip
\noindent\emph{Properties.}
A \emph{property} over finfinite
(\resp infinite) traces is a subset of
$\fTrc$
(\resp of $\Act^\omega$). Simply \emph{property} refers to a finfinite property, unless stated otherwise.
We say that a finite trace $\ftV$ \emph{positively (\resp negatively) determines} a property $\propV \subseteq \fTrc$
when
 $\ftV\fftV \in \propV$ (\resp $\ftV\fftV \notin \propV$), for every $\fftV \in \fTrc$. The same terms apply similarly when $\propV \subseteq \Act^\omega$.
%
We
call
a
property
regular if it is the union of a regular finite property $\propV_f {\in} \Act^*$ and 
an $\omega$-regular infinite property $\propV_i{\in} \Act^{\omega}$.

%

\section{A Monitor-Oriented Hierarchy}
\label{sec:mon-hierarchy}


From a tool-construction perspective, it is important to give concrete, implementable definitions
of
monitors; we do
so
in \Cref{sec:instantiation-regular}. To understand the guarantees that these monitors will provide, we first discuss the
general notion of monitor and monitoring system. Already in this general setting, we are able to identify the various requirements that give rise to the hierarchy of monitorability, depicted
in the middle part
of \Cref{fig:hierarchy}. \Cref{sec:instantiation-regular} will then provide operational semantics to this hierarchy, in the setting of regular properties.

%
We consider a monitor to be an entity that analyses finite traces and (at the very least) identifies a set of finfinite traces that it \emph{accepts} and a set of finfinite traces that it \emph{rejects}. We consider two postulates. Firstly, an acceptance or rejection verdict has to be based on a finite prefix of a trace, \Cref{def:monitoring-system}.1; secondly, verdicts must be \emph{irrevocable}, \Cref{def:monitoring-system}.2. These postulates make explicit two features shared by most monitorability definitions in the literature.


\begin{definition}\label{def:monitoring-system}
	A \emph{monitoring system} is a triple $(M,\accrel,\rejrel)$, where $M$ is a nonempty set of monitors,  $\accrel,\rejrel \subseteq M {\times} \fTrc$, and
	for every $\mV \in M$:
	\begin{enumerate}
		\item $ \bigl(\,\acc{\mV,\fftV} \text{ implies } \exists \ftV \cdot \bigl(\ftV \preceq \fftV \text{ and } \acc{\mV,\ftV}\bigr)\,\bigr)$  and
		$\bigl(\,\rej{\mV,\fftV} \text{ implies } \exists \ftV \cdot \bigl(\ftV \preceq \fftV \text{ and } \rej{\mV,\ftV}\bigr)\,\bigr)$;
		\item $\bigl(\acc{\mV,\ftV} \text{ implies }\forall \fftV {\cdot} \acc{\mV,\ftV\fftV}\bigr)$ and
		$\bigl(\rej{\mV,\ftV} \text{ implies } \forall \fftV {\cdot} \rej{\mV,\ftV\fftV}\bigr)$.
		\qedd
	\end{enumerate}
\end{definition}

\begin{remark}
	Finite automata do not satisfy the requirements of
	\Cref{def:monitoring-system}
	since their judgement can be revoked.
	Standard B\"{u}chi automata
	are
	not
	good candidates either, since they need to read the entire infinite trace to accept or reject.
	\qedd
\end{remark}


%


We define a notion of maximal monitoring system for a collection of properties; for each property $\propV$ in that set, such a system must contain a monitor that reaches a verdict for all traces that have some prefix that determines $\propV$.


\begin{definition}\label{def:abstract-monitoring-system}
	A monitoring system $(M,\accrel,\rejrel)$ is \emph{maximal} for a collection of properties $C \subseteq 2^{\fTrc}$ if for every $\propV \in C$
	there is a monitor $\mV_\propV \in M$ such that
	$(i)$ $\acc{\mV_\propV,\fftV}$  iff trace $\fftV$ has a prefix that positively  determines  $\propV$;   $(ii)$ $\rej{\mV_\propV,\fftV}$ iff trace $\fftV$ has a prefix that negatively determines  $\propV$.
	\qedd
\end{definition}

In \Cref{sec:instantiation-regular}, we present an instance of such a maximal monitoring system for regular properties.
This shows that, for regular properties at least, the maximality of a monitoring system is a reasonable requirement.
Unless otherwise stated, we assume a fixed maximal
 monitoring system $(M,\accrel,\rejrel)$ throughout the rest of the paper.
%
%
%
%
%
For $\mV {\in} M$ to monitor for a
property \propV, it needs to satisfy some requirements.
The most important such requirement is \emph{soundness}.

\begin{definition}[Soundness]\label{def:soundness}
	Monitor $\mV$ is \emph{sound} for property \propV when
	$\acc{\mV,\fftV}$ implies $\fftV \in \propV$,
	and
	$\rej{\mV,\fftV}$ implies $\fftV \notin \propV$.
\exqed
\end{definition}

\begin{lemma}\label{lem:verd-determines}
If $\mV$ is sound for $\propV$ and $\acc{\mV,\ftV}$ (\resp $\rej{\mV,\ftV}$), then $\ftV$ positively (\resp negatively) determines $\propV$.
\qedmaybe
\end{lemma}

\begin{lemma}\label{thm:abstract-is-sound} For every $\propV \subseteq \fTrc$:
		$(i)$
		$\mV_\propV$ is sound for $\propV$; and
		 $(ii)$
		if
		\mV is a sound monitor for \propV
		and
		$\acc{\mV,\fftV}$ (\resp $\rej{\mV,\fftV}$), then it is also the case that
		$\acc{\mV_\propV,\fftV)}$ (\resp $\rej{\mV_\propV,\fftV)}$).
		\qedmaybe
\end{lemma}

The dual requirement to soundness, \ie \emph{completeness}, entails that the monitor detects \emph{all} violating and satisfying traces.
%
%
Unfortunately, this is only possible for trivial properties in the finfinite\footnote{In the infinite domain more properties \emph{are} completely monitorable, see \Cref{sec:other}.}  domain---see \Cref{prop:no-complete}.
Instead, monitors may be required to accept all satisfying traces, or reject all violating traces.

\begin{definition}[Completeness]\label{def:completeness}
	Monitor $\mV$ is \emph{satisfaction-complete} for \propV if
	$\fftV {\in} \propV$ implies
	$\acc{\mV,\fftV}$
	and
	\emph{violation-complete} for \propV if
	$\fftV {\notin} P$ implies
	$\rej{\mV,\fftV}$.
	%
	It is \emph{complete} for \propV if it is \emph{both} satisfaction- and violation-complete for \propV and \emph{partially-complete} if it is \emph{either} satisfaction- \emph{or} violation-complete.
	\qedd
\end{definition}

\begin{proposition}\label{prop:no-complete}
	If \mV is sound and complete for  \propV then
	$\propV
	 {=}
	\fTrc$ or $\propV
	{=}
	\emptyset$.
\end{proposition}
 \begin{proof}
If $\varepsilon \in P$, then
 	$\acc{\mV,\varepsilon}$,
 	so
	from
 	\Cref{def:monitoring-system},
 	$\forall \fftV \in \fTrc.~\acc{\mV,\fftV}$.
 	Due to
 	the soundness of $\mV$, $P = \fTrc$.
 	Similarly,  $P = \emptyset$ when $\varepsilon \notin P$.
 \end{proof}

We define monitorability in terms of the guarantees that the monitors are expected to give. Soundness is
not negotiable.
Given the
consequences
of requiring completeness, as evidenced by \Cref{prop:no-complete}, we consider weaker forms of completeness.
The weaker the completeness guarantee, the more properties can be monitored.

\begin{definition}[Complete Monitorability]\label{def:oper-monitorability-popl}
Property \propV\ is completely monitorable when there is a monitor that is sound and complete for \propV.
It is
\emph{monitorable for satisfactions} (\resp \emph{violations})
when
there is a monitor $\mV$ that
is sound and
satisfaction- (\resp and violation-) complete for \propV.
%
It is \emph{partially} monitorable
when it is monitorable for satisfactions \emph{or}
violations.

A class of properties $\clsV \subseteq 2^\fTrc$ is satisfaction-, violation-, partially, or completely monitorable,
when \emph{every} property $\propV {\in} \clsV$ is, respectively, satisfaction-, violation, partially or completely monitorable.
%
We denote
the  class of all satisfaction, violation, partially, and completely monitorable properties by maximal monitoring systems
as
\clsSC, \clsVC,  \clsPC, and  \clsCC, respectively.
 \qedd
\end{definition}

Since even partial monitorability, the weakest form in \Cref{def:oper-monitorability-popl}, renders a substantial number of properties unmonitorable \cite{AcetoAFIL19}, one may consider even weaker forms of completeness that only flag a
\emph{subset} of satisfying (or violating) traces.  \clsS denotes monitorability \emph{without} completeness requirements.
Arguably, however, the weakest guarantee for a sound monitor of a property \propV to be of use is if it flags at least \emph{one} trace.
One may then further strengthen this requirement and demand that this guarantee is invariant throughout the analysis of a monitor.


\begin{definition}[\Useful Monitors\footnote{These are not related to the \emph{informative prefixes} from~\cite{KV:2001} nor \textit{persistence} from~\cite{rosu2007safety}}] \label{def:informative-monitors}
	Monitor \mV is satisfaction- (\resp violation-) \emph{\useful}
	if $\exists \fftV
	\cdot \acc{\mV,\fftV}$ (\resp \rej{\mV,\fftV}).
	It is satisfaction- (\resp violation-)  \emph{\superuseful} if $\forall \ftV
	\exists \fftV \cdot \acc{\mV,\ftV\fftV}$ (\resp \rej{\mV,\ftV\fftV}).
  We simply say that \mV is \useful (\resp \superuseful) when we do not distinguish between satisfactions or violations. \qedd
\end{definition}


\begin{definition}[\Useful Monitorability] \label{def:informative-monitorability}
	Property \propV is \usefully (\resp \superusefully) monitorable if there is an \useful (\resp a \superuseful) monitor that is sound for \propV .
%
	A class of properties $\clsV {\subseteq} 2^\fTrc$ is \usefully (\resp \superusefully) monitorable,
	when all its properties are \usefully (\resp \superusefully) monitorable.
	The  class of all \usefully (\resp \superusefully) monitorable properties by maximal monitoring systems is denoted as \clsUC (\resp \clsSUC).
A property \propV is \superusefully monitorable for satisfaction (\resp for violation) if there is a  satisfaction- (\resp violation-) \superuseful monitor that is sound for \propV.
We revisit this definition in \Cref{sec:instantiation-regular}.
%
%
%
	\qedd
\end{definition}

\begin{example} \label{ex:partial-monitorability-limits}
	The  property \emph{``\fa never occurs and eventually \su is reached''} (\Cref{ex:intro}) is \emph{not} partially monitorable but \emph{is} \superusefully monitorable.

	The property
	requiring that \emph{``\re only appears a finite number of times"} is \emph{not} \usefully monitorable.
%
	For if it were, the respective sound \useful monitor \mV
	should at least accept or reject one trace.
	If it accepts a trace $\fftV$, by \Cref{def:monitoring-system}, it must accept some prefix $\ftV\preceq\fftV$.
Again, by \Cref{def:monitoring-system}, all continuations, including $\ftV\re^\omega$, must be accepted by \mV. This makes it unsound, which is a contradiction.
%
Similarly, if \mV rejects some $\fftV$, it must reject some finite $\ftV\preceq\fftV$ that necessarily contains a finite number of \re actions,
making it unsound. \qedd
\end{example}

\begin{theorem}[Monitorability Hierarchy] \label{thm:mon-hierarchy}
	The
	monitorability classes given in
	\Cref{def:oper-monitorability-popl,def:informative-monitorability} form the inclusion hierarchy depicted in \Cref{fig:hierarchy}.
\end{theorem}
\begin{proof}
	The hardest inclusion to show
	is $\clsPC = \clsSC{\cup}\clsVC \subseteq \clsSUC$.
%
  %
	Pick a property $\propV \in \clsVC$.
%
	Let $\ftV \in \Act^*$.
	If $\exists \fftV \cdot \ftV\fftV \notin \propV$ then by \Cref{def:completeness} we have $\rej{\mV_\propV,\ftV\fftV}$.
	Otherwise, $\forall \fftV \cdot \ftV\fftV \in \propV$, meaning that \ftV\ positively determines \propV, and by \Cref{def:abstract-monitoring-system} we have  $\acc{\mV_\propV,\ftV\fftV}$.
	By \Cref{def:informative-monitors}, we deduce that $\mV_\propV$ is \superuseful since $\forall \ftV \exists \fftV \cdot \acc{\mV_\propV,\ftV\fftV} \text{ or } \rej{\mV_\propV,\ftV\fftV})$. Thus, by
	\Cref{def:informative-monitorability}, it follows that
	 $\propV \in \clsSUC$.
	 The case for $\propV \in \clsSC$ is dual.
\end{proof}



\section{An Instantiation for Regular Properties}
\label{sec:instantiation-regular}



We provide a concrete maximal monitoring system
for regular properties
. This monitoring system
that gives an operational interpretation to the levels of the monitorability hierarchy, and enables us to find syntactic characterisations for them in \UHML~\cite{Larsen:90:HMLRec,AcetoAFIL19}. Since this logic is a reformulation of the \UCalc~\cite{Koz:83:TCS},
it is expressive enough to describe all regular properties and to embed specification formalisms such as LTL, {($\omega$-)regular} expressions, B\"{u}chi automata, and Street automata, used in the state of the art on monitorability.

\subsectionmaybe{The Logic.}


\begin{figure}[!t]
  \begin{align*}
  \hV,\hVV \in \UHML &\bnfdef  \hTru  
  &
  &\bnfsepp  \hFls &
  & \bnfsepp \hOr{\hV\,}{\,\hVV}  &
  & \bnfsepp \hAnd{\hV\,}{\,\hVV}  
  \\
  &~~~\bnfsepp \hSuf{\acta}{\hV} &
  &\bnfsepp \hNec{\acta}{\hV} &
  & \bnfsepp \hMinX{\hV} &
  & \bnfsepp \hMaxX{\hV} &
  & \bnfsepp\; X 
  \end{align*}
     \[\begin{array}{rlrl}
      \hSem{\hTru,\sigma}  & \deftxt   \fTrc
      &
      \hSem{\hFls,\sigma}  & \deftxt   \emptyset
      \\
      \hSem{\hOr{\hV_1}{\hV_2},\sigma} & \deftxt   \hSem{\hV_1,\sigma} \cup \hSem{\hV_2,\sigma}
      \qquad\qquad
       &
      \hSem{\hAnd{\hV_1}{\hV_2},\sigma} & \deftxt   \hSem{\hV_1,\sigma} \cap \hSem{\hV_2,\sigma}
      \\
    \hSem{\hNec{\acta}{\hV},\sigma}  &
\deftxt \sset{\fftV \;|\;
	\fftV=\acta\fftVV
	\;\text{ implies }\;\fftVV \in \hSem{\hV,\sigma}
}
&
 \hSem{\hSuf{\acta}{\hV},\sigma}  &
	\deftxt \sset{\acta\fftV \;|\; \fftV \in \hSem{\hV,\sigma}
}
%
    \\
    \hSem{\hMin{\!\hVarX}{\hV},\sigma} & \deftxt
    \bigcap \sset{\FSet \;|\;  \hSem{\hV,\sigma[\hVarX\mapsto \FSet]} \subseteq \FSet\ }

    \\
    \hSem{\hMax{\!\hVarX}{\hV},\sigma} & \deftxt
    \bigcup \sset{\FSet \;|\;  \FSet \subseteq \hSem{\hV,\sigma[\hVarX\mapsto \FSet]}\ }
    \quad\;
    &
    \hSem{\hVarX,\sigma} & \deftxt   \sigma(\hVarX)
    \\
    \\
    \end{array}
    \]
  \caption{\UHML Syntax and (finfinite) Linear-Time Semantics}
  \label{fig:recHML}
\end{figure}

The syntax or \UHML is defined by the grammar in \Cref{fig:recHML}, which assumes a countable set of logical variables $\hVarX,\hVarY \in \LVars$.
Apart from the standard constructs for truth, falsehood, conjunction and disjunction, the logic is equipped with
existential (\hSuf{\acta}{\hV}) and universal  (\hNec{\acta}{\hV}) modal operators, and
\emph{two} recursion operators expressing least and greatest fixpoints (\resp \hMinX{\hV} and  \hMaxX{\hV}).
The semantics is given by the function $\hSem{-}$ defined  in \Cref{fig:recHML}.
It maps a (possibly open) formula to a set of (finfinite) traces~\cite{AcetoAFIL19} by induction on the formula structure, using valuations that map logical variables to sets of traces,
$\sigma: \LVars \to \powset{\fTrc}$, where $\sigma(\hVarX)$ is the set of traces assumed to satisfy \hVarX.
An existential modality \hSuf{\acta}{\hV} denotes all traces with a prefix action $\acta$ and a continuation that satisfies \hV whereas a universal modality  \hNec{\acta}{\hV} denotes all traces that are either \emph{not} prefixed by  \acta  or have a continuation \fftVV\ satisfying \hV.
The sets of traces satisfying the least and greatest fixpoint formulae, \hMinX{\hV} and  \hMaxX{\hV}, are the least and the greatest fixpoints, respectively, of the function induced by the formula \hV.
%
%
For closed formulae, we use \hSem{\hV}  in lieu of \hSem{\hV,\sigma} (for some $\sigma$). 
Formulae are generally assumed to be closed and guarded~\cite{kupferman00}.
In the discussions we occasionally treat formulae, $\hV$,  as the properties they denote, $\hSemF{\hV}$.

\begin{example}
  \label{ex:rechml-ltl}
The characteristic LTL operators
can be encoded in \UHML as:
\[
	\begin{array}{rlrlrlr}
	\ltlX\,\hV
  & {\deftxt} \bigvee_{\acta\in\Act}\hSuf{\acta}{\hV}
  \quad\quad
	&
	\hV\,\ltlU\,\hVV
  & {\deftxt} \hMinY{\bigl(\hOr{\hVV\,}{\,(\hAnd{\,\hV\,}{\,\ltlX\ \hVarY})}\bigr)}
  \quad\quad
  & \ltlF\, \hV
  & {\deftxt} \hTru \,\ltlU\, \hV
	\\

	\hV\,\ltlR\,\hVV &
  \multicolumn{3}{l}{
  {\deftxt} \hMaxY{\bigl(\hOr{(\hAnd{\,\hVV\,}{\,\hV\,})\,}{\,(\hAnd{\,\hVV\,}{\,\ltlX\ \hVarY})}\bigr)}
   }
  & \ltlG\, \hV
  & {\deftxt} \hFls \,\ltlR\, \hV
\end{array}
\]
In examples, atomic propositions \acta and $\neg\acta$ \resp denote $\hSuf{\acta}\hTru$ and $\hNec{\acta}\hFls$.
\exqed
\end{example}

For
better readability,
examples use LTL.
Since we operate in the finfinite domain, $\ltlX$ should be read as a \emph{strong} next operator,
in line with
\Cref{ex:rechml-ltl}.

%
%


\begin{figure}[!t]
	\begin{align*}
	\mV,\mVV\in\Mon\ &\bnfdef\  \vV
	&& \bnfsep \prf{\acta}{\mV}
	&& \bnfsep \ch{\mV}{\mVV}
  && \bnfsep \mV \paralC \mVV
  && \bnfsep \mV \paralD \mVV
	&& \bnfsep \rec{x}{\mV}
	&&\bnfsep x \\
	\vV,\vVV \in \Verd & \bnfdef\ \stp
	&&  \bnfsep \no
	&& \bnfsep \yes
	\end{align*}
	\begin{mathpar}
		\inference[\rtit{mAct}]{
		}{\prf{\acta}{\mV}  \traSS{\acta} \mV}
    \and
    \inference[\rtit{mVer}]{
    }{\vV \traSS{\acta} \vV}
		\and
		%
		\inference[\rtit{mRec}]{\mV\subS{\rec{x}{\mV}}{x} \traSS{\acta} \mVV}{\rec{x}{\mV} \traSS{\acta} \mVV}
		\\\\\\
		\inference[\rtit{mSelL}]{\mV \traSS{\acta} \mV'}{\ch{\mV}{\mVV}  \traSS{\acta} \mV'}
		%
		%
    \and
    \inference[\rtit{mPar}]{\mV\traSS{\acta}\mV' & \mVV\traSS{\acta}\mVV'}
    {\mV\paralG\mVV\traSS{\acta}\mV'\paralG\mVV'}
    \\\\\\
    \inference[\rtit{mTauL}]{\mV\traSS{\tau}\mV'}
    {\mV\paralG\mVV\traSS{\tau}\mV'\paralG\mVV}
    \quad
     \inference[\rtit{mVrE}]{}{\stp\paralG\stp \traS{\tau} \stp}
     \quad
    \inference[\rtit{mVrC1}]{}{{{\yes\paralC\mV} \traS{\tau} {\mV}}}
    \\\\\\
    \inference[\rtit{mVrC2}]{}{{{\no\paralC\mV} \traS{\tau} {\no}}}
   \and
    \inference[\rtit{mVrD1}]{}{{{\no\paralD\mV} \traS{\tau} {\mV}}}
    \and
    \inference[\rtit{mVrD2}]{}{{{\yes\paralD\mV} \traS{\tau} {\yes}}}
	\end{mathpar}
	\caption{Monitor Syntax and Labelled-Transition Semantics}
	\label{fig:monit-instr}
\end{figure}

\subsectionmaybe{The Monitors.}

We consider the operational monitoring system of \cite{FraAI:17:FMSD,AcetoAFIL19}, summarised in \Cref{fig:monit-instr} (symmetric rules for binary operators are omitted).
The full system is given in \Cref{sec:regular-defs-appendix}.
Monitors are states of a transition system where \ch{\mV}{\mVV} denotes an (external) choice and $\mV \paralG \mVV$ denotes a composite monitor where $\paralG \in \sset{\paralD,\paralC}$.
There are three distinct \emph{verdict} states, \yes, \no, and \mend, although only the first two are relevant to monitorability.
This semantics gives an operational account of how a monitor in state \mV incrementally analyses a sequence of actions $\ftV=\acta_{1}\ldots\acta_{k}$ to reach a new state \mVV;
%
the monitor \mV accepts (\resp rejects) a trace \fftV, \acc{\mV,\fftV} (\resp \rej{\mV,\fftV}), when it can transition to the verdict state \yes\ (\resp \no) while analysing a prefix $\ftV\preceq\fftV$.
Since verdicts are  irrevocable (rule \rtit{mVer} in \Cref{fig:monit-instr}),
it is not hard to see that this operational framework satisfies the conditions for a monitoring system of \Cref{def:monitoring-system}.
%
%
The monitoring system of \Cref{fig:monit-instr} is also maximal for regular properties, according to \Cref{def:abstract-monitoring-system}.
This concrete instance thus demonstrates the realisability of the abstract definitions in \Cref{sec:mon-hierarchy}.

\begin{theorem}\label{thm:determines-gives-verdict}
	For all $\hV {\in} \UHML$, there is a monitor  $\mV{\in}\Mon$ that is sound for $\hV$ and accepts
	all finite traces
	that positively determine $\hV$
	and  rejects
	all finite traces
	that negatively determine $\hV$.
	\qedmaybe
\end{theorem}

As a corollary of \Cref{thm:determines-gives-verdict}, from \Cref{lem:verd-determines} we deduce that for any arbitrary monitoring system $(M,\accrel,\rejrel)$,
if $\mV \in M$ is sound for some $\hV\in\UHML$, then there is a monitor $\mVV\in\Mon$ from \Cref{fig:monit-instr} that accepts (\resp rejects) all traces \fftV\ that $\mV$ accepts (\resp rejects).
In the sequel, we thus assume that the fixed monitoring system is $(\Mon,\accrel,\rejrel)$ of \Cref{fig:monit-instr}, as it subsumes all others.


\section{A Syntactic Characterisation of Monitorability}
\label{sec:syntactic}

We present syntactic  characterizations for the various monitorability classes as fragments of \UHML.

\subsectionmaybe{Partial Monitorability, syntactically.}

In~\cite{AcetoAFIL19} Aceto \etal
identify a maximal partially monitorable syntactic fragment of \UHML.



\begin{theorem}[Partial Monitorability \cite{AcetoAFIL19}]\label{thm:partial-monitorability} \label{def:shml-chml} Consider the
  fragments:
  \begin{align*}
	\hV,\hVV \in \SHML &\bnfdef \hTru \bnfsepp \hFls \bnfsepp \hNec{\acta}{\hV} \bnfsepp \hAnd{\hV}{\hVV} \bnfsepp \hMaxX{\hV} \bnfsepp \hVarX \text{ and }\\
	\hV,\hVV \in \CHML &\bnfdef \hTru \bnfsepp \hFls \bnfsepp \hSuf{\acta} \hV \bnfsepp \hOr{\hV}{\hVV} \bnfsepp \hMinX{\hV} \bnfsepp \hVarX .
	\end{align*}
	The fragment \SHML is monitorable for violation whereas \CHML is monitorable for satisfaction.
	Furthermore, if $\hV \in \UHML$ is monitorable
	for satisfaction (\resp for violation) by some $\mV{\in}\Mon$, then it is expressible in
	$\CHML$(\resp \SHML) , \ie   $\exists \hVV {\in} \CHML$ 	(\resp $\hVV {\in} \SHML$), such that $\hSemF{\hV} {=} \hSemF{\hVV}$.
	\qedmaybe
\end{theorem}




As a corollary of \Cref{thm:partial-monitorability}, any $\hV \in \UHML$ that is monitorable for satisfaction (\resp for violation)
can also be expressed as some $\hVV \in \CHML$ (\resp $\hVV \in \SHML$) where $\hSemF{\hV} = \hSemF{\hVV}$. For this fragment, the following automated synthesis function, which is readily implementable, is given in~\cite{AcetoAFIL19}.
\[\begin{array}{rlrlrl}
  \hSyn{\hFls} & \defeq \no \quad&
  \hSyn{\hAndF} &\defeq \hSyn{\hV_1} \paralC \hSyn{\hV_2} \quad&
  \hSyn{\hMaxXF} & \defeq  \rec{x}{\hSyn{\hV}}
  \\
  \hSyn{\hTru} & \defeq \yes &
  \hSyn{\hOrF} &\defeq \hSyn{\hV_1} \paralD \hSyn{\hV_2}
  &
  \hSyn{\hMinXF} & \defeq \rec{x}{\hSyn{\hV}}
  \\

    \hSyn{\hNec{\acta}{\hV}} &
    \multicolumn{3}{l}{
    \textstyle
    \defeq
     \ch{\prf{\acta}{\hSyn{\hV}}}{\sum_{\actb\in\Act\setminus\sset{\acta}}\prf{\actb}{\yes}}}
   &
   \hSyn{\hVarX} &\defeq x
   \\
	\hSyn{\hSuf{\acta} \hV} &
  \multicolumn{3}{l}{
  \textstyle
  \defeq \ch{\prf{\acta}{\hSyn{\hV}}}{\sum_{\actb\in\Act\setminus\sset{\acta}}\prf{\actb}{\no}}
  }
\end{array}\]


\subsectionmaybe{\Useful Monitorability, syntactically.}

We proceed to identify syntactic fragments of \UHML that correspond to \useful monitorability.


\begin{definition}\label{def:informative-fragments}
	The \useful  fragment is  $\IHML = \SIHML \cup \CIHML$ where
	\begin{align*}
	\SIHML &= \{ \hV_1 \land \hV_2 \in \UHML \mid
	\hV_1 \in \SHML
	\text{ and }
	\hFls \text{ appears in } \hV_1
	\},\\
	\CIHML &= \{ \hV_1 \lor \hV_2 \in \UHML \mid
	\hV_1 \in \CHML
	\text{ and }
	\hTru \text{ appears in } \hV_1
	\}
\tag*{\qedd}
	\end{align*}
\end{definition}

\begin{theorem}\label{thm:useful-is-ihml}
	For $\hV \in \UHML$,
	$\hV$ is \usefully monitorable if and only if
	there is some $\hVV \in \IHML$, such that $\hSemF{\hVV} = \hSemF{\hV}$.
	\qedmaybe
\end{theorem}

\begin{example}\label{ex:useful}
 $\ltlG\,\neg \fa \land \ltlF\, \su$
	from \Cref{ex:partial-monitorability-limits} (expressed here in LTL) is a $\SIHML$ property, as $\ltlG \neg \fa$ can be written in \SHML as
  $\hMaxX{\hAnd{\hNec{\fa}\hFls}{\hAnd{\hNec{\su}{\hVarX}}{\hNec{\re}{\hVarX}}}}$.
	In contrast, $\ltlF\ltlG \neg \re$ cannot be written in \IHML, as it is not \usefully monitorable.
	\qedd
\end{example}

\begin{remark}
	In $\SIHML$ and $\CIHML$, $\hV_1$ describes an \useful part of the formula, that is, a formula with at least one path to \hTru (or \hFls), which indicates that the corresponding finite trace determines the property.
	Monitor synthesis from these fragments can use this part of the formula to synthesize a monitor
	that detects the finite traces that satisfy (violate) $\hV_1$.
	The value of the synthesised monitor then depends on $\hV_1$.
	It is therefore important to have techniques to extract some $\hV_1$ that will retain as much monitoring information as possible. This extraction is outside the scope of this paper and left as future work.
	\qedd
\end{remark}

\subsectionmaybe{Persistently \Useful Monitorability, syntactically.}

We also give a syntactic characterization of the \UHML properties that are \superusefully monitorable for satisfaction or violation.
As the requirements for \superuseful monitors are
more subtle than for \useful monitors, the fragments we present are equally more involved than those for \useful monitorability.

\begin{definition}\label{def:active-reactive-formulae}
	We define $\eHML$, the explicit fragment of \UHML:
	\begin{align*}
	\hV \in \eHML &\bnfdef  \hTru&
	&\bnfsepp  \hFls
	&  
	& \bnfsepp \hMinX{\hV}  &
	& \bnfsepp \hMaxX{\hV}  
	& \bnfsepp\; X
	\\
	& \bnfsepp \hOr{\hV\,}{\,\hVV}
	&& \bnfsepp \hAnd{\hV\,}{\,\hVV}
	&&\bnfsepp \bigvee_{\act \in \Act}\hSuf{\act}{\hV_\act}
	&&
	\bnfsepp \bigwedge_{\act \in \Act}\hNec{\act}{\hV_\act}
	. \tag*{\begin{tabular}{r}
		\\
		\qedd
		\end{tabular}}
	\end{align*}
\end{definition}

\begin{example}
	Formula $\hNec{\fa}{\hNec{\su}\hFls}$ is not explicit, but it can be rewritten as the explicit formula
  $\hNec{\fa}{(\hNec{\su}{\hFls} \land \hNec{\fa}{\hTru} \land \hNec{\re}{\hTru})} \land \hNec{\su}\hTru \land \hNec{\re}\hTru$.
	\qedd
\end{example}

Roughly, the following definition captures whether \hTru and \hFls are reacheable from subformulae (where the binding of a variable is reachable from the variable).

\begin{definition}\label{def:refute-formulae}
	Given a closed $\SHML$ (\resp \CHML) formula $\hV$, we   define for a subformula $\hVV$ that it can refute (\resp verify) in 0 unfoldings, when $\hFls$ (\resp $\hTru$) appears in $\hVV$, and that it can refute (\resp verify) in $k+1$ unfoldings, when it can refute (\resp verify) in $k$ unfoldings, or $X$ appears in $\hVV$ and $\hVV$ is in the scope of a subformula $\max X.\hVV'$ (\resp  $\min X.\hVV'$) that can refute (\resp verify) in $k$ unfoldings.
	We simply say that $\hVV$ can refute (\resp verify) when it can refute (\resp verify) in $k$ unfoldings, for some $k \geq 0$.
	\qedd
\end{definition}

\begin{example}
	For formula $\hMaxX{\hNec{\su}{X} \land \hNec{\fa} \hFls \land \hNec{\re} \hFls}$, subformula $\hNec{\su}{X} \land \hNec{\fa} \hFls \land \hNec{\re} \hFls$ can refute in $0$ unfoldings. In contrast,
	$\hNec{\su}{X}$ cannot refute in $0$ unfoldings, but it can refute in 1, because $X$ appears in it and $\hMaxX{\hNec{\su}{X} \land \hNec{\fa} \hFls \land \hNec{\re} \hFls}$ can refute in $0$ unfoldings. Therefore, all subformulae of
  $\hMaxX{\hNec{\su}{X} \land \hNec{\fa} \hFls \land \hNec{\re} \hFls}$ can refute.
	\qedd
\end{example}

We now define the
fragments of \UHML corresponding to \UHML properties that are \superusefully monitorable for satisfaction or violation.

\begin{definition}\label{def:p-informative-fragments}
	We define the fragment
	$\PIHML = \SPIHML \cup \CPIHML$ where:
	\begin{align*}
	\SPIHML &= \left\{ \hV_1 \land \hV_2 \in \UHML \Bigm|
	\begin{tabular}{l}
	$\hV_1 \in \SHML\cap \eHML$ 
	and every
	\\
	subformula of $\hV_1$ can refute
	\end{tabular}
	\right\}\\
	\CPIHML &= \left\{ \hV_1 \lor \hV_2 \in \UHML \Bigm|
	\begin{tabular}{l}
	$\hV_1 \in \CHML \cap \eHML$ 
	and every \\
	subformula of $\hV_1$ can verify
	\end{tabular}
	\right\}
  \tag*{
  $\begin{array}{r}
     \\[0.73em]
    \qedd
  \end{array}$
  }
	\end{align*}
\end{definition}

\begin{theorem}\label{thm:supuseful-regular-is-pihml}
	For $\hV \in \UHML$,
	$\hV$ is \superusefully monitorable for violation (\resp for satisfaction) if and only if there is some $\hVV \in \SPIHML$ (\resp $\hVV \in \CPIHML$), such that $\hSemF{\hVV} = \hSemF{\hV}$.
	\qedmaybe
\end{theorem}

\begin{remark}
	To the best of our efforts, a syntactic characterisation of \superuseful monitorability would involve pairs of equivalent formulae with parts from $\SHML$ and $\CHML$ that together become, in some sense, explicit.
	We leave such a characterization as future work.
	\qedd
\end{remark}

%


%
%

\section{Safety and Co-safety}
\label{sec:safety}

The classic (and perhaps the most intuitive)
definition of monitorability consists of (some variation of) \emph{safety} properties~\cite{alpern1985defining,schneider2000enforceable,viswanathan2004foundations,FalconeFernandezMounier:STTT:12,AcetoAFIL19}.
Nevertheless, there are subtleties associated with how exactly safety properties are defined---particularly over the finfinite domain---and how decidable they need to be to qualify as truly monitorable.
For example, Kim and Viswanathan~\cite{viswanathan2004foundations} argued that only recursively enumerable safety properties are monitorable (they restrict themselves to infinite, rather than finfinite traces).
By and large, however,
most works on monitorability restrict themselves to regular properties, as we
do in \Cref{sec:instantiation-regular}.

We adopt the  definition of safety that is intuitive for the context of RV: a property can be considered monitorable if its failures can be identified by a finite prefix.
This is equivalent to Falcone \etal's formal definition of safety properties\cite[Def.~4]{FalconeFernandezMounier:STTT:12} and
work such as~\cite{alpern1985defining,ChangMannaPnueli:92:ALP} when restricted to infinite traces.

\begin{definition}[Safety]\label{def:safety}
A property $\propV \subseteq \fTrc$ is a \emph{safety property} if every $\fftV \notin \propV$ has a prefix, \ftV\ that determines $\propV$ negatively. The class of safety properties is denoted as \clsSafe in \Cref{fig:hierarchy}. \qedd
\end{definition}

 Pnueli and Zaks, and Falcone \etal (among others) argue that it makes sense to monitor both for violation and satisfaction.
 Hence, if safety is monitorable for violations, then the dual class, co-safety (\aka guarantee~\cite{FalconeFernandezMounier:STTT:12}, reachability~\cite{berard2013systems}), is monitorable for satisfaction. That is, every trace that satisfies a co-safety property can be positively determined by a finite prefix.

\begin{definition}[Co-safety]\label{def:co-safety}
A property $\propV \subseteq \fTrc$ is a \emph{co-safety property} if every $\fftV \in \propV$ has prefix, \ftV, that determines $\propV$ positively. The class of safety properties is denoted as \clsCSafe, also represented in \Cref{fig:hierarchy}. \qedd
\end{definition}

\begin{example}\label{ex:safety-cosafety} \emph{``Eventually $\su$ is reached''}, \ie  $\ltlF\,\su$, is a co-safety property whereas \emph{``$\fa$ never occurs''}, \ie $\ltlG\, \neg\fa$, is a safety property.
The property \emph{``$\su$ occurs infinitely often''}, \ie $\ltlG\,\ltlF\,\su$, is neither safety nor co-safety. The property only holds over infinite traces so it cannot be positively determined by a finite trace.  Dually, there is \emph{no} finite trace that determines that there cannot be an infinite number of \su occurrences in a continuation of the trace.
\exqed
\end{example}

\paragraph{Safety and Co-safety, operationally.}
%
It should come as no surprise that safety and co-safety coincide with an equally natural operational definition.
Here, we establish the correspondence with the denotational definition of safety (co-safety), completing three correspondences amongst the monitorability classes of \Cref{fig:hierarchy}

\begin{theorem}\label{thm:safety-oper} $\clsVC=\clsSafe$ and  $\clsSC=\clsCSafe$.
\end{theorem}
\begin{proof}
	We treat the case for safety, as the case for co-safety is similar.
	If $\propV$ is a safety property, then for every $\fftV \in \fTrc \setminus \propV$, there is some finite prefix $\ftV$ of $\fftV$ that negatively determines $\propV$.
	Therefore, $\mV_\propV$ is sound (\Cref{thm:abstract-is-sound}) and violation-complete (\Cref{def:abstract-monitoring-system}) for $\propV$.
	The other direction follows from \Cref{lem:part-mon-gives-goodbad}.
\end{proof}

Aceto \etal~\cite{AcetoAFIL19} already show the correspondence between violation (dually, satisfaction) monitorability over finfinite traces and properties expressible in \SHML (dually, \CHML).
As a corollary of \Cref{thm:safety-oper}, we obtain a syntactic characterisation for the \clsSafe and \clsCSafe monitorability classes.

\begin{remark}\label{rem:wtf-falcone}
Falcone \etal~\cite[Def. 17, Thm. 3]{FalconeFernandezMounier:STTT:12} propose definitions of monitorability over finfinite traces that are claimed to coincide with the classes \clsSafe, \clsCSafe and their union. However, this claim is incorrect. The properties, ``the trace is finite"  and $\ltlG\,\ltlF\,\su$ from \Cref{ex:safety-cosafety} are neither safety nor co-safety properties.  On the other hand, they are monitorable according to the alternative monitorability definition given in \cite[Def. 17]{FalconeFernandezMounier:STTT:12}. If the results claimed in \cite[Thm. 3]{FalconeFernandezMounier:STTT:12} held true, this would contradict the fact that those properties are neither safety nor co-safety properties. See \Cref{sec:falcone-appendix} for further details.
\exqed
\end{remark}

%
%
%
%

\section{Pnueli and Zaks}
\label{sec:pz}


The work on monitorability due to Pnueli and Zaks~\cite{PnueliZaks:06:FM} is often cited by the RV community~\cite{BartocciFFR:18:RVIntro}.
The often overlooked particularity of their definitions is that they only define monitorability of a property \textit{with respect to a (finite) sequence}.
%
%

\begin{definition}[\cite{PnueliZaks:06:FM}]\label{def:s-monitorable}
Property $\propV$ is $\ftV$-monitorable,
where
$\ftV\in \Act^*$, if there is some $\ftVV\in \Act^*$ such that $\propV$ is positively or negatively determined by $\ftV\ftVV$. \qedd
\end{definition}

\begin{example}
The property $\bigl(\fa
\land
\ltlF\, \re
\bigr) \vee \bigl(\ltlF\, \ltlG\, \su\bigr)$ is $\ftV$-monitorable for any finite trace that begins with \fa, \ie $\fa\ftV$, since it is determined by the extension $\fa\ftV\re$.
It is \emph{not} $\ftV$-monitorable for finite traces that
begin with an action other than $\fa$.
\exqed
\end{example}

Monitorability over properties---rather than over property--sequence pairs---can then be defined by either quantifying \textit{universally} or \textit{existentially} over 	finite traces: a property is monitorable either if it is \ftV-monitorable for all \ftV, or for some \ftV.
 We address both definitions, which we call \upz - and \epz -monitorability respectively.
 \upz -monitorability is the more standard interpretation: it appears for example in~\cite{FalconeFernandezMounier:STTT:12,BauerLeuckerSchallhart:10:LandC} where it is attributed to Pnueli and Zaks.
 However, the original intent seems to align more with \epz -monitorability: in~\cite{PnueliZaks:06:FM}, Pnueli and Zaks refer to a property as non-monitorable if it is not monitorable for \textit{any} sequence.
 This interpretation coincides with \textit{weak monitorability} used in~\cite{Chen18weak}.

\begin{definition}[\upz -monitorability]\label{def:upz-monitorable}
	A property $\propV$ is (universally Pnueli--Zaks) \upz -monitorable if it is \ftV-monitorable for \emph{all} finite
	traces
	\ftV.  The class of all \upz -monitorable properties is denoted \clsUPZ. 
	\qedd
\end{definition}

\begin{definition}[\epz -monitorability]\label{def:epz-monitorable}
	A property is (existentially Pnueli--Zaks) \epz -monitorable if it is \ftV-monitorable for \emph{some} finite
	trace
	\ftV, \ie if it is $\varepsilon$-monitorable. The class of \epz -monitorable properties is written \clsEPZ. 
	\qedd
\end{definition}


The apparently innocuous choice between existential and universal quantification leads to different monitorability classes \clsUPZ and \clsEPZ .

\begin{example}\label{ex:epz-upz}
%
%
Consider the property \emph{``Either \su occurs before \fa, or  \re happens infinitely often"}, expressed in LTL fashion as $\bigl((\neg \fa) \,\ltlU\, \su\bigr) \vee \bigl(\ltlG\, \ltlF\,\re\bigr)$.
This property is \epz -monitorable because
every finite trace $\ftV \fa$
positively determines the property.
However, it is \emph{not} \upz -monitorable
%
%
because no extension of the trace \fa\ positively or negatively determines that property. Indeed, all extensions of \fa\ violate the first disjunct and, as we argued in \Cref{ex:safety-cosafety}, there is no finite trace that determines the second conjunct  positively or negatively.
%
\exqed
\end{example}

From \Cref{def:upz-monitorable,def:epz-monitorable}, it
follows immediately
that $\clsUPZ \subset\clsEPZ$.

\begin{proposition}\label{prop:safety-to-upz}
All properties in $\clsSafe \cup \clsCSafe$ are \upz
-monitorable.
\end{proposition}
\begin{proof}
	Let $\propV \in \clsSafe$ and
	pick a finite trace \ftV.
	If there is an \fftV\ such that $\ftV\fftV \notin \propV$ then, by \Cref{def:safety},
	there exists $\ftVV \preceq \ftV\fftV$ that negatively determines $\propV$, meaning that \ftV\ has an extension that negatively determines $\propV$.
	Alternatively, if there is \emph{no} $\fftV$ such that $\ftV\fftV \notin \propV$, $\ftV$ itself positively determines $\propV$.
	Hence $\propV$ is  \ftV-monitorable, for \emph{every} \ftV, according to \Cref{def:s-monitorable}.
	The case for $\propV \in \clsCSafe$ is dual.
	%
\end{proof}

\paragraph*{Pnueli and Zaks, operationally.}
%
\epz -monitorability coincides with
informative monitorability: \epz -monitorable properties are those for which some monitor can reach a verdict on some finite trace.
For similar reasons, \upz -monitorability coincides with
\superuseful monitorability. See \Cref{fig:hierarchy}.

\begin{theorem}\label{thm:epz} \label{thm:upz-iff-superuseful} \label{thm:epz-iff-useful}
$\clsEPZ = \clsUC$ and $\clsUPZ = \clsSUC$.
%
\end{theorem}

\begin{proof}
Since the proofs of the two claims are analogous, we simply outline the one
for $\clsUPZ = \clsSUC$.
%
%
Let $\propV\in\clsUPZ$ and pick a finite trace $\ftV\in\Act^*$.
By \Cref{thm:abstract-is-sound}, $\mV_\propV$ is sound for $\propV$.
By \Cref{def:informative-monitors} we need to show that there exists an \fftV\ such that $\acc{\mV_\propV,\ftV\fftV}$ or  $\rej{\mV_\propV,\ftV\fftV}$.
From \Cref{def:upz-monitorable,def:s-monitorable} we know that there is a finite \ftVV\ such that $\ftV\ftVV$ positively or negatively determines \propV.
By \Cref{def:abstract-monitoring-system} we know that $\acc{\mV_\propV,\ftV\ftVV}$ or  $\rej{\mV_\propV,\ftV\ftVV}$.
Thus $\propV\in\clsSUC$, which is the required result.

Conversely, assume $\propV\in\clsSUC$, and pick a
$\ftV \in \Act^*$.
By \Cref{def:upz-monitorable,def:s-monitorable}, we need to show that there is an extension of \ftV\ that
positively or negatively determines \propV.
From \Cref{def:informative-monitors,def:informative-monitorability}, there exists a \fftV\ such that $\acc{\mV_\propV,\ftV\fftV}$ or $\rej{\mV_\propV,\ftV\fftV}$.
By \Cref{def:monitoring-system}, there is a finite extension of \ftV, say $\ftV\ftVV$, that is a prefix of $\ftV\fftV$ such that  $\acc{\mV_\propV,\ftV\ftVV}$ or $\rej{\mV_\propV,\ftV\ftVV}$.
 By \Cref{def:abstract-monitoring-system}, we know that $\ftV\ftVV$ either positively or negatively determines \propV.
 Thus $\propV \in \clsUPZ$.
%
%
%
%
%
%
%
%
%
\end{proof}

\section{Monitorability in other settings}
\label{sec:other}

We have shown how classical definitions of monitorability 
fit into our hierarchy and provided the corresponding operational interpretations and syntactic characterisations,
%
focussing on regular finfinite properties over a finite alphabet and monitors with irrevocable verdicts.
Here we discuss how different parameters, both within our setting and beyond, affect what is monitorable.

\paragraph{Monitorability \wrt the alphabet.}
	The monitorability of a property can depend on $\Act$. For instance, if $\Act$ has at least two elements \sset{\acta,\actb,\ldots}, property $\{\acta^\omega \}$, which can be represented as $\mx X.\diam{\acta}X$, is $s$-monitorable for every sequence $s$, as $s$ can be extended to $s\actb$, which negatively determines the property.
	On the other hand,
	assume that $\Act = \{\acta\}$. In this case, $\{\acta^\omega \}$ is neither \epz - nor \upz -monitorable. Indeed, no string $s=a^k$, $k \geq 0$, determines $\{\acta^\omega \}$ positively or negatively as $s$ does not satisfy $p$ but its extension $a^\omega$ does.
	On the other hand, when restricted to infinite traces, $p$ is again \epz -monitorable.

So far, we only considered finite alphabets; how an infinite alphabet, which may encode integer data for example, affects monitorability is left as future work.

\paragraph{Monitoring with revocable verdicts.} Early on, we postulated that verdicts are irrevocable. Although this is a typical (implicit) assumption in most work on monitorability, some authors have considered monitors that give revocable judgements when an irrevocable one is not appropriate. This approach is taken by Bauer \etal when they define a finite-trace semantics for LTL, called RV-LTL~\cite{BauerLeuckerSchallhart:10:LandC}.
 Falcone \etal~\cite{FalconeFernandezMounier:STTT:12} also have a definition of monitorability based on this idea (in addition to those discussed in \Cref{rem:wtf-falcone}). It uses the four-valued domain $\{\yes, \no, \yes_c, \no_c\}$ ($c$ for \emph{currently}). Finite traces that do not determine a property yield
a (revocable) verdict $\yes_c$ or $\no_c$ that indicates whether the trace observed so far  satisfies the property; \yes\ and \no\ are still irrevocable. This definition allows \emph{all} finfinite properties to be monitored since it does not require verdicts to be irrevocable.

This type of monitoring does not give any guarantees beyond soundness: there are properties that are monitorable according to this definition for which no sound monitor ever reaches an irrevocable verdict: $\ltlF ~\ltlG~ \su$ for the system from \Cref{ex:intro} has no sound informative monitor, yet can be monitored according to Falcone \etal 's four-valued monitoring. This type of monitorability is complete, in the sense of providing at least a \emph{revocable} verdict for all traces.

\paragraph{Monitorability in the infinite and finite.}
Bauer \etal use \upz -monitorability in their study of runtime verification for LTL
~\cite{Bauer:2011} and attribute it to Pnueli and Zaks. However, unlike Falcone \etal, Pnueli and Zaks~\cite{PnueliZaks:06:FM} and ourselves, they focus on properties over \emph{infinite}
traces. There are some striking differences that arise if there is no risk of an execution ending.
Aceto \etal show that, unlike in the finfinite domain, a set of non-trivial properties becomes completely monitorable: HML \cite{HennessyM:1985:JACM} (\aka modal logic) is both satisfaction- and violation-monitorable over infinite traces~\cite{AcetoAFIL19}. Furthermore, some properties, like $\{\acta^\omega\}$ over $\Act = \{\acta\}$, that were not \epz - or \upz -monitorable on the finfinite domain, are \epz- or even \upz -monitorable on the infinite domain. The full analysis of how the hierarchy in \Cref{fig:hierarchy} changes for the infinite domain is left for future work.

Barringer \etal~\cite{barringer2008rule}
consider monitoring of properties over  \emph{finite} traces.
In this domain, all properties are monitorable if, as is the case in~\cite{barringer2008rule}, the end of a trace is \emph{observable}; in this setting the question of monitorability is less relevant.

\section{Conclusion}
\label{sec:conclusion}

We have proposed a unified, operational view on monitorability.
This allows us to clearly state the implicit operational guarantees of existing definitions of monitorability. For instance, recall \Cref{ex:intro} from the introduction:  since $(\ltlG\,\neg\fa)\wedge(\ltlF\,\su)$  is \epz - and \upz -monitorable but it is not a safety nor co-safety property, we know
there is a monitor which can recognise some violations and satisfactions of this property, but there is no monitor that can recognise \emph{all} satisfactions nor \emph{all} violations.

%
%

Although we
focussed on the setting of regular, finfinite properties, the definitions of monitorability in \Cref{sec:mon-hierarchy}, and, more fundamentally, the methodology that systematically puts the relationship between monitor behaviour and specification centre stage, are equally applicable to other settings. We have already mentioned the infinite domain and richer alphabets in \Cref{sec:other}. Another
interesting direction would be to lift the restriction to regular properties and finite-state monitors. Indeed, while Barringer \etal consider a specification logic that allows for context-free properties~\cite{barringer2008rule}, in~\cite{FerrereHS:18:LICS}, Ferrier \etal consider monitors with registers (\ie infinite state monitors) to verify safety properties that are not regular. It seems likely that monitorability beyond safety may also be worth studying in these extended settings.

\section*{Acknowledgements}
  This research was partially supported by the  projects  ``TheoFoMon: Theoretical Foundations for Monitorability'' (grant number: 163406-051) and ``Epistemic Logic for Distributed Runtime Monitoring'' (grant number: {184940-051}) of the Icelandic Research Fund, by the BMBF project ``Aramis II'' (project number:{01IS160253}) and the EPSRC project ``Solving parity games in theory and practice'' (project number:{EP/P020909/1}).


\bibliographystyle{plain}
\bibliography{refs2}

\newpage

\appendix


\Cref{sec:falcone-appendix} discusses in more details Falcone \etal 's claim that their definitions of monitorability coincide with safety and co-safety properties and provides counter-examples showing this to be incorrect.
\Cref{sec:regular-defs-appendix} gives further technical details regarding the monitoring system presented in \Cref{sec:instantiation-regular}. \Cref{sec:proofs-appendix} contains all the omitted proofs.

\section{Monitorability \`a la Falcone \etal}
\label{sec:falcone-appendix}

Falcone \etal~\cite{FalconeFernandezMounier:STTT:12} propose three definitions of monitorability (Definitions 16 and 17 in~\cite{FalconeFernandezMounier:STTT:12}) which they claim to coincide with safety, co-safety, and the union of safety and co-safety properties (Theorem 3 in \cite{FalconeFernandezMounier:STTT:12}). 
We discuss this claim in more detail here, and argue that it does not hold. In brief, their definition deems all properties that are uniform over finite traces, such as ``$\succeed$ infinitely often", or ``the trace is finite" to be monitorable, not just safety and co-safety properties. In this appendix we recall Falcone \etal 's definitions and show that their definitions of monitorability include more than just safety and co-safety properties.

\begin{remark} Falcone \etal present finfinite properties as a pair consisting of a set of finite traces and a set of infinite traces. Here we will speak of just one set, containing both finite and infinite traces.
\end{remark}

The definition of monitorability 
proposed by Falcone \etal in \cite{FalconeFernandezMounier:STTT:12}
is parameterised by a truth domain, and a mapping of formulas into this domain. They then give a uniform condition that defines monitorability with respect to any truth-domain and its associated mapping. Here we focus on their monitorability with respect to the truth-domains $\{\hTru,?\}$, $\{\hFls,?\}$ and $\{\hTru,\hFls,?\}$, which they claim correspond to co-safety, safety and their union, respectively.

\begin{definition}[Property evaluation with respect to a truth-domain~\cite{FalconeFernandezMounier:STTT:12}]
For each of three different verdict-domains and finfinite properties $\propV$ (``$r$-properties'' in their terminology), Falcone, Fernandez and Mournier define the following evaluation functions:
\begin{description}
\item[ For $\tdom=\{\hFls, ?\}$ and $\ftV\in \Act^*$:]~\\
$\hSem{\propV}_\tdom (\ftV)= \hFls $ if $\forall \fftV \in \fTrc.~ \ftV\fftV\notin \propV$
\\
$\hSem{\propV}_\tdom (\ftV)= ? $ otherwise.
\item[ For $\tdom=\{\hTru, ?\}$ and $\ftV\in \Act^*$:]~\\
$\hSem{\propV}_\tdom (\ftV)= \hTru $ if $\forall \fftV\in \fTrc.~ \ftV\fftV\in \propV$
\\
$\hSem{\propV}_\tdom (\ftV)= ? $ otherwise.

\item[For $\tdom= \{\hTru, \hFls, ?\}$ and $\ftV\in \Act^*$:]~\\
$\hSem{\propV}_\tdom (\ftV)= \hTru $ if $\ftV\in \propV$ and $\forall \fftV \in \fTrc.~ \ftV\fftV\in \propV$
\\
$\hSem{\propV}_\tdom (\ftV)= \hFls $ if $\ftV\notin \propV$ and $\forall \fftV\in \fTrc.~ \ftV\fftV\notin \propV$
\\
$\hSem{\propV}_\tdom (s)= ? $ otherwise.
\qedd 
\end{description}
\end{definition}
\begin{definition}[\ffm -monitorability~Definition 17, \cite{FalconeFernandezMounier:STTT:12}]
A property $\propV$ is $\tdom$-monitorable over a truth domain $\tdom$ if for all $\ftV,\ftVV\in \Act^*$, if $\ftV\in \propV$ and $\ftVV\notin \propV$, then
$\hSem{\propV}_\tdom (\ftV)\neq \hSem{\propV}_\tdom(\ftVV)$.
\qedd 
\end{definition}

From this definition, it easily follows that any property $\propV$ for which $\propV \cap \Act^* = \emptyset$ or $\Act^*\subseteq\propV$ is vacuously monitorable for any truth-domain, and evaluation function. However, not all such properties are safety or co-safety properties: ``always eventually $\succeed$'' for instance is neither a safety nor a co-safety property.

We believe the critical points are Lemma 3 and Theorem 3 in \cite{FalconeFernandezMounier:STTT:12}, which do not hold. The proof of Lemma 3 in particular (Appendix 2.3) falsely claims that  $\propV \cap \Act^*=\emptyset$ or $\Act^*\subseteq \propV$ implies that $\propV$ is a safety or co-safety properties.

\section{An Operational Monitoring System: Regular Monitors}
\label{sec:regular-defs-appendix}

The monitoring system 
$(\Mon,\accrel,\rejrel)$
is given by a \emph{Labelled Transition System (LTS)} based on $\Act$,  which is comprised of the monitor states, or monitors, and a transition relation.
%
%
%
%
The set of monitor states,
$\Mon$, and the monitor transition relation, $\reduc \subseteq (\Mon\times(\Act\cup\sset{\tau})\times \Mon)$, are defined in \Cref{fig:monit-instr}.
%
%
The suggestive notation $\mV \traS{\acttt} \mVV$ denotes $(\mV,\acttt,\mVV) \in \reduc$; we also write $\mV \traSN{\acttt}$ to denote $\neg(\exists \mVV. \;\mV\traS{\acttt}\mVV)$.
We employ the usual notation for weak transitions and write $\mV \wreduc \mVV$ in lieu of $\mV (\traS{\tau})^{\ast} \mVV$ and $\mV \wtraS{\acttt} \mVV$ for
$\mV \wreduc\cdot\traS{\acttt}\cdot\wreduc \mVV$.
We
write  sequences of transitions
$\mV\wtra{\act_1}\cdots\wtra{\act_k} \mVV$ as $\mV \wtraS{\ftr} \mVV$,
where $\ftr = \act_1\cdots\act_k$.
A monitor that does not use any parallel operator is called a \emph{regular monitor}.
The full monitoring system and regular monitors were defined and used
 in \cite{AceAFI:17:FSTTCS,FraAI:17:FMSD,AcetoAFIL19}.

%
%
%
%

\begin{definition}[Acceptance and Rejection]  \label{def:acc-n-rej}
	For a monitor $\mV \in \Mon$, we define $\rej{\mV,\ftV}$ (\resp $\acc{\mV,\ftV}$)
	and say that
	\emph{\mV rejects} (\resp \emph{accepts})
	when
	$\mV \wtraS{\ftV} \no$
	(\resp
	$\mV \wtraS{\ftV} \yes$).
	Similarly, 
	for $\tV \in \Act^\omega$, we write 
	$\rej{\mV,\tV}$ (\resp $\acc{\mV,\tV}$)
	if there exist $ \ftV \in \Act^* $ and $\tVV\in\Act^\omega$  such that
	$ \tV = \ftV\tVV$ and  \mV rejects (\resp accepts) \ftV.
	%
	\qedd
\end{definition}

For a finite nonempty set of indices $I$, we use
$\sum_{i \in I}\mV_i$ to denote any combination of the monitors in
$\{\mV_i \mid i \in I\}$
using the operator $+$.
%
%
For each $j \in I$, $\sum_{i \in I}\mV_i$ is called a sum of $\mV_j$, and $\mV_j$ is called a summand of $\sum_{i \in I}\mV_i$.
The following \Cref{lem:ver-persistence} assures us that regular monitors satisfy the conditions to be a monitoring system, given in \Cref{def:monitoring-system}.

\begin{lemma}[Verdict Persistence, \cite{FraAI:17:FMSD,AcetoAFIL19}]\label{lem:ver-persistence}
	\begin{math}
	\vV \wtraS{\ftV} \mV \text{ implies } \mV=\vV.
	\end{math}\qedmaybe
\end{lemma}

We will use the following definitions and results in the proofs of \Cref{sec:proofs-appendix}.
We 
define determinism for regular monitors.

\begin{definition}[\cite{AceAFI:2017:CIAA,determinization}]\label{def:determinism}
	A closed regular monitor $\mV$ is \emph{
		deterministic} iff
	every sum of at least two summands
	that appears in $\mV$
	is of the form
	$\sum_{\act \in A} \act.m_\act$,
	where $A \subseteq \Act$.
	\qedd
\end{definition}

\begin{definition}[Verdict Equivalence]\label{def:verd-eq}
	Monitors $\mV$ and $\mVV$ are called \emph{verdict equivalent} when for every $\fftV \in \fTrc$, $\acc{\mV,\fftV}$ iff  $\acc{\mVV,\fftV}$
	and
	$\rej{\mV,\fftV}$ iff  $\rej{\mVV,\fftV}$.
	\qedd 
\end{definition}

\begin{theorem}[\cite{AceAFI:2017:CIAA,AcetoAFIL19}]\label{thm:determinization}
	Every monitor in $\Mon$ is verdict equivalent to a deterministic regular monitor.
	\qedmaybe 
\end{theorem}

In the following, we
say that 
$\propV$ is 
suffix-closed when for all $\ftV,\ftVV \in \Act^*$, $\ftV \in \propV$ implies $\ftV\ftVV \in \propV$ --- notice that we only quantify over finite traces.
The suffix-closure of $\propV$ is $\{ \ftV\fftV \in \fTrc \mid \ftV \in \propV \}$.

\begin{theorem}[\cite{AceAFI:2017:CIAA,determinization}]\label{thm:regular-to-regular}
	If $L,L' \subseteq \Act^*$ are regular and 
	suffix-closed, 
	and $L \cap L' = \emptyset$, 	
	then there is a regular monitor	$\mV$, such that $\acc{\mV,\ftV}$ iff $\ftV \in L$ and 
	$\rej{\mV,\ftV}$ iff $\ftV \in L'$.
	\qedmaybe 
\end{theorem}

In this appendix, we use the formula synthesis function from regular monitors to \SHML formulae, from \cite{FraAI:17:FMSD,AcetoAFIL19}:
\begin{align*}
\mSyn{\no} &= \hFls 
&\mSyn{\stp} &= \mSyn{\yes} = \hTru
& \mSyn{x} &= X
\\
\mSyn{\esel{\mV}{\mVV}} &= \hAnd{f(\mV)}{\mSyn{\mVV}} 
&\mSyn{\prf{\acta}\mV} &= \hNec{\acta}{\mSyn{\mV}}  &\mSyn{\recX{\mV}} &= \hMaxX{\mSyn{\mV}}
\end{align*}

\begin{theorem}[\cite{AcetoAFIL19}]
	For every regular monitor $\mV$, 
	$\mSyn{\mV} \in \SHML$, and 
	$\mV$ is sound and violation-complete for $\mSyn{\mV}$.
\qedmaybe
\end{theorem}

\section{Proofs Omitted from the Main Document}
\label{sec:proofs-appendix}

Here we present the proofs of results  that were omitted from the main text.

We use the following classical result (see \cite{arnold2001rudiments} for more on the $\mu$-calculus and \UHML):

\begin{lemma}\label{lem:regularMu}
	If $\hV \in \UHML$, then $\hSemF{\hV}\cap\Act^*$ is regular.
	\qedmaybe 
\end{lemma}

Due to \Cref{thm:determinization}, we can assume that every monitor in $\Mon$ is a regular, or deterministic regular monitor. We often do so in the following proofs.

\begin{definition}\label{def:dplus-dminus}
	Let $\propV \subseteq \fTrc$.
	We define the following two sets of finite traces:

	$D^-_{\propV} = \{\ftV \in \Act^* \mid \propV \text{ is negatively determined by } \ftV\}$;

	$D^+_\propV = \{\ftV \in \Act^* \mid \propV \text{ is positively determined by } \ftV\}$.
	\qedd
\end{definition}

\begin{lemma}\label{lem:part-mon-gives-goodbad}
	If $\propV \subseteq \fTrc$ is monitorable for satisfaction (\resp for violation) by any monitoring system, then every $\fftV \in \propV$ (\resp $\fftV \in \fTrc \setminus \propV$) has a finite prefix that positively (\resp negatively) determines $\propV$.
\end{lemma}

\begin{proof}
	We treat the case for satisfaction, as the case for violation is dual.
	Let $\fftV \in \propV$ and $\mV$ be a monitor that is sound and satisfaction-complete for $\propV$.
	Then, due to satisfaction-completeness, $\acc{\mV,\fftV}$, and by the requirements of \Cref{def:monitoring-system}, there is a finite prefix $\ftV$ of $\fftV$, such that $\acc{\mV,\ftV}$. Therefore, by the same requirements, for every $\fftVV\in \fTrc$, $\acc{\mV,\ftV \fftVV}$. As we know that $\mV$ is sound for $\propV$, this yields that $\ftV$ positively determines $\propV$.
\end{proof}

For a set of finite traces $S \subseteq \Act^*$, we define
$$\min S = \{ \ftV \in S \mid \forall \ftVV \in S.~ (\ftVV \text{ is a prefix of }\ftV \Rightarrow \ftV = \ftVV)  \}.$$

\begin{lemma}\label{lem:AAL1}
	Assume
	that every finfinite trace that satisfies (\resp violates) $\propV$ has a prefix that positively (\resp negatively) determines \propV.
	Then $\propV$ (\resp $\fTrc \setminus \propV$) is the suffix-closure of $\min (\propV \cap \Act^*)$ (\resp of $\min ((\fTrc \setminus \propV) \cap \Act^*)$).
\end{lemma}
\begin{proof}
	Again, we only treat the case for satisfaction.
	Let $\fftV \in \propV$.
	By our assumptions, there is at least one finite prefix
	of $\fftV$ that positively determines $\propV$.
	These prefixes of $\fftV$ are well-ordered by the prefix relation, and therefore there is a smallest prefix $\ftV$ of $\fftV$ that positively determines $\propV$.
	Therefore, $\ftV \in \min (\propV \cap \Act^*)$, and we see that $\fftV$ is in the suffix-closure of $\min (\propV \cap \Act^*)$.
	Conversely, let $\ftV \in \min (\propV \cap \Act^*)$ and $\fftV$ be an extension of $\ftV$.
	Then, $\ftV \in \hSem{\propV}$, so by the proviso of the lemma, there is a prefix $\ftVV$ of $\ftV$ (and of $\fftV$) that positively determines $\propV$,\footnote{The reader may also notice that $\ftV = \ftVV$, due to the minimality of $\ftV$.} and therefore $\fftV \in \propV$.
\end{proof}


\subsubsection{Proofs Omitted from \Cref{sec:instantiation-regular}: Regularity and Monitors}

\begin{lemma}\label{lem:regularMuD}
	If $\hV \in \UHML$, then $D^+_{\hV}$ and $D^-_{\hV}$ are regular.
\end{lemma}
\begin{proof}
	We know that $\hSemF{\hV}\cap \Act^*$ is regular (\Cref{lem:regularMu}) and $\hSemF{\hV}\cap \Act^\omega$, the infinite-trace interpretation of $\hV$, is $\omega$-regular.
	Therefore, there are a DFA $D_F$ that recognizes
	$\hSemF{\hV}\cap \Act^*$
	and a deterministic $\omega$-automaton $D_I$ that
	recognizes
	$\hSemF{\hV}\cap \Act^\omega$.
	Let $A_F = \{ \ftV \in \Act^* \mid \forall \ftVV \in \Act^*.~ \ftV\ftVV \in \hSemF{\hV} \}$
	and
	$A_I = \{ \ftV \in \Act^* \mid \forall \tV\in \Act^\omega.~ \ftV\tV \in \hSemF{\hV} \}$.
%
%
	Let $Q_F$ (\resp $Q_I$) be the set of states 
	in $D_F$ (\resp in $D_I$) 
	that can be reached reading some trace $\ftV \in \Act^*$.
	By construction, for each $\ftV \in \Act^*$, we have that
%
	$\ftV \in A_F$ (\resp $\ftV \in A_I$) if and only if $\ftV$ ends in the respective set in the automaton.
	Therefore, there are DFAs $D'_F$ and $D'_I$ for $A_F$ and $A_I$, respectively, and thus $D^+_{\hV} = A_F \cap A_I$ is regular.
	The case for $D^-_{\hV}$ is similar.
\end{proof}

\begin{thmbiss}{\ref{thm:determines-gives-verdict}}
For all $\hV {\in} \UHML$, there is a monitor  $\mV{\in}\Mon$ that is sound for $\hV$ and accepts 
all finite traces 
that positively determine $\hV$  
and  rejects
all finite traces 
that negatively determine $\hV$.  
\end{thmbiss}

\begin{proof}
	By \Cref{lem:regularMuD}, $D^+_{\hV}$ and $D^-_{\hV}$, the  sets of finite traces that (respectively) positively or
	negatively determine $\hV$ are regular.
	It is also not hard to see that they are 
	suffix-closed.
	Therefore the theorem follows from \Cref{thm:regular-to-regular}.
\end{proof}

\begin{corollary}\label{thm:anything-goes}
	If $\mV$ is a sound monitor for $\hV \in \UHML$, then there is a regular monitor $\mVV$ that is sound for $\hV$, and such that for every $\ftV \in \Act^*$, $\acc{\mV,\ftV}$ implies $\acc{\mVV,\ftV}$, and $\rej{\mV,\ftV}$ implies $\rej{\mVV,\ftV}$.
\end{corollary}

\begin{proof}
	By \Cref{thm:determines-gives-verdict}, there is a regular monitor $\mVV$ that is sound for $\hV$, and accepts all finite traces that positively determine $\hV$, and rejects all the finite traces that negatively determine $\hV$.
	If $\acc{\mV,\ftV}$ (\resp $\rej{\mV,\ftV}$) for some finite trace $\ftV$, then, due to the soundness of $\mV$, $\ftV \in \hSemF{\hV}$ (\resp $\ftV \notin \hSemF{\hV}$), and therefore, from \Cref{lem:verd-determines}, $\ftV$ positively (\resp negatively) determines $\hV$. By the properties of $\mVV$, we have that $\acc{\mVV,\ftV}$ (\resp $\rej{\mVV,\ftV}$). 
\end{proof}

\subsubsection{Proofs Omitted from \Cref{sec:syntactic}: Syntactic Characterizations}

\begin{definition}\label{def:refute-monitors}
	Let $\mV$ be a closed regular monitor, 
	and let $\mVV$ be one of its submonitors.
	We say that
	$\mVV$ 
	can refute (\resp verify) in 0 unfoldings, when $\no$ (\resp \yes) appears in $\mVV$, and that it can refute (\resp verify) in $k+1$ unfoldings, when it can refute (\resp verify) in $k$ unfoldings, or $x$ appears in $\mVV$ and $\mVV$ is in the scope of a submonitor $\rec x \mVV'$ of $\mV$ that can refute (\resp verify) in $k$ unfoldings.
	We simply say that $\mVV$ can refute (\resp verify) in $\mV$ when it can refute (\resp verify) in $k$ unfoldings, for some $k \geq 0$.
	\qedd
\end{definition}

We define the depth of $\hFls$ in an \SHML formula in a recursive way: $d_{\hFls}(\hFls) = 0$; $d_{\hFls}(\hTru) = d_{\hFls}(X) = \infty$; $d_{\hFls}(\hVV_1 \land \hVV_2) = \min\{d_{\hFls}(\hVV_1), d_{\hFls}(\hVV_2)\}+1$; $d_{\hFls}(\hNec{\act}{\hVV}) = d_{\hFls}(\hVV) + 1$; and $d_{\hFls}(\max X.\hVV) = d_{\hFls}(\hVV) + 1$.

\begin{lemma}\label{lem:substit-lowers-depth}
	For all possibly open $\hV, \hVV \in \SHML$ 
	$d_{\hFls}(\hV[\hVV/X]) \leq d_{\hFls}(\hV)$.
\end{lemma}
\begin{proof}
	Straightforward induction on $\hV$.
\end{proof}

\begin{lemma}\label{thm:ihml-to-useful-regular}
	If $\hV \in \IHML$, then there is a regular monitor that is sound and \useful for \hV.
\end{lemma}
\begin{proof}
	We assume that $\hV \in \SIHML$, as the case for $\hV \in \CIHML$ is similar.
	Let $\hV = \hV_1 \land \hV_2$, where $\hV_1 \in \SHML$ and $\hFls$ appears in $\hV_1$.
	We prove by induction on
	$d_{\hFls}(\hVV)$ that for all $\hVV \in \SHML$, if $d_{\hFls}(\hVV)<\infty$, then 
	there is a finite trace that negatively determines $\hVV$. If $d_{\hFls}(\hVV) = 0$, then
	$\hVV = \hFls$, and we are done, as $\varepsilon$ negatively determines $\hFls$.
	Otherwise, $d_{\hFls}(\hVV) = k+1$ and we consider the following cases:
	\begin{description}
		\item[$\hVV = \hVV_1 \land \hVV_2$] 
		In this case, either $d_{\hFls}(\hVV_1) = k$ or $d_{\hFls}(\hVV_2) = k$, so by the inductive hypothesis, there is a finite trace that negatively determines one of the two conjuncts, and therefore also $\hVV$.
		\item[$\hVV = \hNec{\act} \hVV'$]
		In this case, $d_{\hFls}(\hVV') = k$, so, by the inductive hypothesis, there is a finite trace $\ftV$ that negatively determines $\hVV'$, so $\act\ftV$ negatively determines $\hVV$.
		\item[$\hVV = \max X. \hVV'$]
		In this case, $d_{\hFls}(\hVV') = k$.
		Therefore, from \Cref{lem:substit-lowers-depth}, $d_{\hFls}(\hVV'[\hVV/X]) \leq d_{\hFls}(\hVV') = k$,
		so, by the inductive hypothesis, there is a finite trace $\ftV$ that negatively determines $\hVV'[\hVV/X]$, so it also negatively determines $\hVV$, because $\hSemF{\hVV'[\hVV/X]} = \hSemF{\hVV}$.
	\end{description}
	As $\hFls$ appears in $\hV_1$, $d_{\hFls}(\hV)<\infty$, so 
	there is a finite trace that negatively determines $\hV_1$, and therefore also $\hV$.
	The theorem follows from \Cref{thm:determines-gives-verdict}.
\end{proof}

\begin{lemma}\label{thm:useful-to-ihml}
	If $\hV \in \UHML$ and there is a monitor that is sound and \useful for \hV, then there is some $\hVV \in \IHML$ such that $\hSemF{\hVV} = \hSemF{\hV}$.
\end{lemma}
\begin{proof}
	If $\mV$ is sound and \useful for $\hV$, then by \Cref{lem:verd-determines}, there is a finite trace $\ftV$ that positively or negatively determines $\hV$. 
	Without loss of generality, we assume that $\ftV$ positively determines $\hV$.
	We can then easily construct a formula $\hVV_1(\ftV)$ that is satisfied
	%
	exactly by $\ftV$ and all its extensions, recursively on $\ftV$: let $\hVV_1(\varepsilon) = \hTru$, and  let $\hVV_1(\act \ftV) = \hSuf{\act}{\hVV_1(\ftV)}$.
	Then, 
	let $\hVV = \hVV_1(\ftV) \lor \hV$. 
	Thus,
	$\hVV \in \IHML$ and $\hSemF{\hVV} = \hSemF{\hV}$.
\end{proof}

\begin{thmbiss}{\ref{thm:useful-is-ihml}}
	For $\hV \in \UHML$, 
	$\hV$ is \usefully monitorable if and only if 
	there is some $\hVV \in \IHML$, such that $\hSemF{\hVV} = \hSemF{\hV}$.
\end{thmbiss}
\begin{proof}
	A consequence of \Cref{thm:ihml-to-useful-regular,thm:useful-to-ihml}.
\end{proof}

\begin{lemma}\label{lem:can-refute-unfold}
	If $\max X.\hV \in \SHML$ can refute in $\max X.\hV \in \SHML$, then it is also the case that $\hV[\max X.\hV/X]$ can refute in $\hV[\max X.\hV/X]$.
	\qedmaybe 
\end{lemma}

\begin{lemma}\label{lem:all-can-refute}
	If all subformulae of $\hNec \act \hV$ or $\hV \land \hVV$ or $\hVV \land \hV$ can refute, then all subformulae of $\hV$ can refute.
	If all subformulae of $\max X.\hV$  can refute, then all subformulae of $\hV[\max X.\hV/X]$ can refute.
	\qedmaybe 
\end{lemma}

We
define the box-depth of a formula from 
$\eHML \cap \SHML$  recursively: 
\begin{align*}
d_{B}\left(\bigwedge_{\actttt \in \Act}\hNec{\actttt}\hV_\actttt\right) &~=~ d_{B}(\hFls) ~=~ 0
\\
d_{B}(X) &~=~ 	d_{B}(\hTru) ~=~ \infty \\
d_{B}(\hV_1 \land \hV_2) &~=~ \min\{d_{B}(\hV_1), d_{B}(\hV_2)\} + 1; ~~~~~\text{ and }\\
d_{B}(\max X.\hV') &~=~ d_{B}(\hV') + 1 .
\end{align*}

\begin{lemma}\label{lem:substit-lowers-depth2}
	For all possibly open $\hV, \hVV \in \eHML\cap\SHML$ 
	$d_{B}(\hV[\hVV/X]) \leq d_{B}(\hV)$.
\end{lemma}
\begin{proof}
	Straightforward induction on $\hV$.
\end{proof}

\begin{lemma}\label{lem:move-violation-forward}
	Let $\act \in \Act$ and $\hV \in \eHML \cap \SHML$ (\resp $\hV \in \eHML \cap \CHML$), where all subformulae of $\hV$ can refute (\resp verify).
	There is some $\hVV \in \eHML \cap \SHML$ (\resp $\hVV \in \eHML \cap \CHML$), such that
	all subformulae of $\hVV$ can refute (\resp verify), and for every $\fftV \in \fTrc$,
	$\act\fftV \in \hSemF{\hV}$ implies that
	$\fftV \in \hSemF{\hVV}$
	(\resp
	$\fftV \in \hSemF{\hVV}$ implies that
	$\act\fftV \in \hSemF{\hV}$).
\end{lemma}

\begin{proof}
	We assume that $\hV \in \eHML \cap \SHML$, as 
	the case for $\hV \in \eHML \cap \CHML$ is similar.
	Since $\hV$ is a closed formula and can refute, $\hFls$ appears in $\hV$, and therefore $d_{B}(\hV) < \infty$.
	We proceed to prove the lemma by induction on
	$d_{B}(\hV)$, similarly to the proof of \Cref{thm:ihml-to-useful-regular}.
	\begin{description}
		\item[If $\hV = \hFls$,] then we are done immediately by taking $\hVV = \hFls$.
		\item[If $\hV = \bigwedge_{\actttt \in \Act}\hNec{\actttt}\hV_\actttt$,] then we can set $\hVV = \hV_\act$.
		\item[If $\hV = \hV_1 \land \hV_2$,] then either  $d_{\acta}(\hV_1) < \infty$ or $d_{\acta}(\hV_2) < \infty$, and we are done by the inductive hypothesis on one of the two subformulae.
		\item[If  $\hV = \max X.\hV'$,] then
		%
		$\hV'[\hV/X] \in \eHML \cap \SHML$ and
		all subformulae of $\hV'[\hV/X]$ can refute, by \Cref{lem:all-can-refute}.
		Furthermore,
		$\hSemF{\hV} = \hSemF{\hV'[\hV/X]}$, and we are done by the inductive hypothesis.
		\qedhere
	\end{description}
\end{proof}

\begin{lemma}\label{thm:pihml-to-supuseful-regular}
	If $\hV \in \SPIHML$ or $\hV \in \CPIHML$, then there is a regular monitor that is sound for $\hV$ and
	\superusefulviol, or, respectively, \superusefulsat.
\end{lemma}
\begin{proof}
	We assume that $\hV \in \SPIHML$, as the case for $\hV \in \CPIHML$ is similar.
	Let $\hV = \hVV \land \hVV_*$, where $\hVV \in \eHML \cap \SHML$
	and all of its subformulae can refute.
	%
	By \Cref{thm:determines-gives-verdict}, i%
	t suffices to prove that for every $\ftV \in \Act^*$, there is some $\ftVV \in \Act^*$, such that
	$\ftV\ftVV$ negatively determines $\hV$.
	%
	We prove this by induction on $\ftV$.
	If $\ftV = \varepsilon$, then as in the proof of \Cref{thm:ihml-to-useful-regular}, we can show that there is a finite trace that negatively determines $\hVV$.
	If $\ftV = \acta\ftV'$,
	then
	by \Cref{lem:move-violation-forward}. there is some $\hVV' \in \eHML \cap \SHML$, such that all subformulae of $\hVV'$ can refute, and
	for every $\fftV \in \fTrc$, $\acta\fftV \in \hSemF{\hVV}$ implies that $\fftV \in \hSemF{\hVV'}$.
	By the inductive hypothesis, there is some $\ftVV$, such that $\ftV'\ftVV$ negatively determines $\hVV'$, and therefore, $\ftV\ftVV$  negatively determines $\hVV$.
\end{proof}

We define the depth of a variable $x$ in a regular monitor $\mV$ recursively: $d_x(x) = 0$ and $d_x(y) = d(\no) = d(\yes) = \infty$, where $y \neq x$; 
$d_{x}(\mV_1 + \mV_2) = \min\{d_{x}(\mV_1), d_{x}(\mV_2)\}+1$; $d_{x}(\act.\mV) = d_{x}(\mV) + 1$; and $d_{x}(\rec x ~\mV) = d_{x}(\rec y ~\mV) = d_{x}(\mV) + 1$.

\newcommand{\actc}{c}

\begin{lemma}\label{lem:explicit-deterministic}
	Let $\mV$ be a \superusefulviol, deterministic regular monitor.
	If $A \subsetneq \Act$, then 
	$\sum_{\act\in A}\act.\mV_\act$ can only appear in $\mV$ as a submonitor of a larger sum.
\end{lemma}

\begin{proof}
	Let $\acta \in \Act \setminus A$ and let $\mV'$ be an open monitor and $x$ a variable that does not appear in $\mV$, such that $\mV = \mV'[\sum_{\act\in A}\act.\mV_\act/x]$. 
	It is clear that $\sum_{\act\in A}\act.\mV_\act \centernot{\wtraS{\acta}}$.
	Therefore, it suffices to
	prove that for every deterministic $\mVV$ with free variable $x$, if $\mVV' \centernot{\wtraS{\acta}}$, then 
	there is a finite trace $\ftV$, such that there is no regular monitor $\mVVV$ for which $\mVV[\mVV'/x] \wtraS{\ftV\acta} \mVVV$.
	We proceed to prove this claim
	by induction on $d_{x}(\mVV)$, and the case for $\mVV = x$ is immediate.
	If $\mVV = \mVV_1 + \mVV_2$, then, as $\mVV$ is deterministic, $\mVV = \actb.\mVV_1' + \actc.\mVV_2'$, where $\actb \neq \actc$, and we are done by the inductive hypothesis on either $\mVV_1'$ or $\mVV_2'$, and $\mVV'$.
	If $\mVV = \actb.\mVV_1$, then if
	the inductive hypothesis on  $\mVV_1'$  and $\mVV'$ gives trace $\ftVV$, then we can set $\ftV =\actb \ftVV$.
	If $\mVV = \rec y \mVV_1$, then we re done by the inductive hypothesis on $\mVV_1[\mVV/y]$ and $\mVV'[\mVV/y]$.
\end{proof}

Here we call a regular monitor explicit when it is generated by the grammar:
\begin{align*}
\mV \bnfdef \mend ~~\bnfsep~~ \no ~~\bnfsep~~ x ~~\bnfsep~~ \sum_{\act \in Act}\act.\mV_\act ~~\bnfsep~~ \rec x \mV.
\end{align*}

\begin{corollary}\label{cor:explicit-deterministic}
	Every \superusefulviol, deterministic regular monitor is explicit.
\end{corollary}
\begin{proof}
	From \Cref{lem:explicit-deterministic}.
\end{proof}

\begin{lemma}\label{lem:ehml-conditions-from-m-to-f}
	Let $\mV$ be an explicit deterministic regular monitor, 
	such that all of its submonitors can refute.
	Then, $\mSyn{\mV} \in \eHML$ and all of its subformulae can refute.
\end{lemma}
\begin{proof}
	By induction on the construction of \mV.
\end{proof}

\begin{lemma}\label{thm:to-supuseful-regular-to-pihml}
	If $\hV \in \UHML$ and there is a monitor that is sound for $\hV$ and
	\superusefulviol or \superusefulsat,
	then there is some $\hVV \in \SPIHML$, or, respectively, $\hVV \in \CPIHML$, such that $\hSemF{\hVV} = \hSemF{\hV}$.
\end{lemma}
\begin{proof}
	We treat the case where the monitor is \superusefulviol, as the case for a \superusefulsat monitor is similar.
	From \Cref{thm:anything-goes}, there is a regular monitor, $\mV$, that is sound for $\hV$ and \superusefulviol.
	%
	By \Cref{thm:determinization},
	we can assume that $\mV$ is deterministic (\Cref{def:determinism}).
	From \Cref{cor:explicit-deterministic},
	$\mV$ 
	is explicit.
	If there is a submonitor of $\mV$ that cannot refute, then we can inductively on $\mV$ prove that there is a finite trace $\ftV$, for which there is no finite trace $\ftVV$, such that $\mV \wtraS{\ftV\ftVV}\no$, which is a contradiction.
	Therefore, from \Cref{lem:ehml-conditions-from-m-to-f},
	the \SHML formula $\mSyn{\mV}$ that we can synthesize from $\mV$ is in $\eHML$, and all of its subformulae can refute.
	Since $\mV$ is sound for $\hV$ and sound and violation complete for $\mSyn{\mV}$, it is the case that $\fTrc\setminus \hSemF{\mSyn{\mV}} \subseteq \fTrc\setminus \hSemF{\hV}$, and therefore $\mSyn{\mV} \land \hV \in \SPIHML$ and $\hSemF{\mSyn{\mV} \land \hV} = \hSemF{\hV}$.
\end{proof}

\begin{thmbiss}{\ref{thm:supuseful-regular-is-pihml}}
	For $\hV \in \UHML$, 
	$\hV$ is \superusefully monitorable for violation (\resp for satisfaction) if and only if there is some $\hVV \in \SPIHML$ (\resp $\hVV \in \CPIHML$), such that $\hSemF{\hVV} = \hSemF{\hV}$.
\end{thmbiss}

\begin{proof}
	A consequence of \Cref{thm:pihml-to-supuseful-regular,thm:to-supuseful-regular-to-pihml}.
\end{proof}

\end{document}